\documentclass[a4paper,twoside,12pt]{article}

\usepackage{graphicx}       % To include graphicx etc.?
\usepackage{amsmath}
\usepackage{amssymb}
\usepackage{color}          % For coloring text
\usepackage[round]{natbib}  % For citation, bibliographystyles etc
\usepackage{enumitem}       % Letters instead of the default arabic numerals in list
\usepackage{url}            % What is the purpose?

\allowdisplaybreaks         % Allows to break eqnarry environment into pages

\newtheorem{theorem}{Theorem}[section]

\newenvironment{proof}[1][Proof]{\begin{trivlist}\item[\hskip \labelsep {\bfseries #1}]}{\end{trivlist}}
\newtheorem{lemma}[theorem]{Lemma}

\newcommand{\conv}[1]{\overset{#1}\rightarrow}
\newcommand{\G}{\mathbb{G}}
\newcommand{\R}{\mathbb{R}}
\newcommand{\compcent}[1]{\vcenter{\hbox{$#1\circ$}}}
\newcommand{\comp}{\mathbin{\mathchoice
  {\compcent\scriptstyle}{\compcent\scriptstyle}
  {\compcent\scriptscriptstyle}{\compcent\scriptscriptstyle}}}

\def\spacingset#1{\renewcommand{\baselinestretch}%
{#1}\small\normalsize} \spacingset{1}

\begin{document}

%    \spacingset{1.45}

    \title{\mbox{Estimating historic movement of a climatological variable}\\from a pair of misaligned data sets\footnote{\textcolor{blue}{Technical report no. ASU/2017/17}; Indian Statistical Institute, Kolkata, India;}}

    \author{Dibyendu Bhaumik \thanks{Corresponding Author; Assistant Adviser, Department of Statistics and Information Management, Reserve Bank of India, C9, 6th Floor, Bandra Kurla Complex, Bandra (East), Mumbai 400051, Maharashtra, India (e-mail: dbhaumik@rbi.org.in).} \and and Debasis Sengupta \thanks{Professor, Applied Statistics Unit, Indian Statistical Institute, 203, Barrackpore Trunk Road, Kolkata 700108, West Bengal, India (e-mail: sdebasis@isical.ac.in).}}

%    \date{July 11, 2017}
    \date{}
    \maketitle

    \begin{abstract}
        We consider in this paper the problem of estimating the mean function from a pair of paleoclimatic functional data sets, after one of them has been registered with the other. We show theoretically that registering one data set with respect to the other is the right way to formulate this problem, which is in contrast with estimation of the mean function in a ``neutral'' time scale that is preferred in the analysis of multiple sets of longitudinal growth data. Once this registration is done, the Nadaraya-Watson estimator of the mean function may be computed from the pooled data. We show that, if a consistent estimator of the time transformation is used for this registration, the above estimator of the mean function would be consistent under a few additional conditions. We study the potential change in asymptotic mean squared error of the estimator that may be possible because of the contribution of the time-transformed data set. After demonstrating through simulation that the additional data can lead to improved estimation in spite of estimation error in registration, we estimate the mean function of three pairs of paleoclimatic data sets. The analysis reveals some interesting aspects of the data sets and the estimation problem.
    \end{abstract}


\section{Introduction}\label{sec:intro}

    Paleoclimatic data on movement of atmospheric concentration of carbon dioxide with time, derived from ice-cores drilled at Lake Vostok and EPICA (The European Project for Ice Coring in Antarctica ) Dome C of Antarctica \citep{Petit_et_al_1999, Luthi_et_al_2008}, show remarkable similarity (see Fig.~\ref{fig:co2_epica_vostok}). The ups and downs of these curves are linked with different phases of the Earth's paleoclimatic history. A more precise description of this movement should be possible by pooling of the two data sets for a combined estimate. However, due to distortion of the time scales arising from errors in radio isotope dating, the two data sets need to be aligned before they are pooled. Numerous techniques for registration are available in the literature, including shape invariant model based registration \citep{Lawton_Sylvestre_Maggio_1972, Kneip_Gasser_1988, Kneip_Engel_1995, Brumback_Lindstrom_2004}, functional principal component based registration \citep{Silverman_1995, Kneip_Ramsay_2008}, dynamic time warping \citep{Wang_Gasser_1997, Wang_Gasser_1999}, continuous monotone registration \citep{Ramsay_Li_1998}, registration by local regression \citep{Kneip_et_al_2000}, maximum likelihood registration through parametric modeling of time transformation \citep{Ronn_2001, Gervini_Gasser_2005}, self-modelling registration \citep{Gervini_Gasser_2004}, functional convex synchronization model based registration \citep{Liu_Muller_2004}, pair-wise curve synchronization \citep{Tang_Muller_2008}, kernel-matched registration \citep{Bhaumik_Srivastava_Sengupta_2017} and so on. It may appear that estimation of the mean function is a rather trivial job, once the data have been aligned by one of the above techniques. We would show in this paper that there are some problems with the conventional wisdom in this matter and set up a clear path to consistent estimation.
    \begin{figure}[h!]
        \begin{center}
            \includegraphics{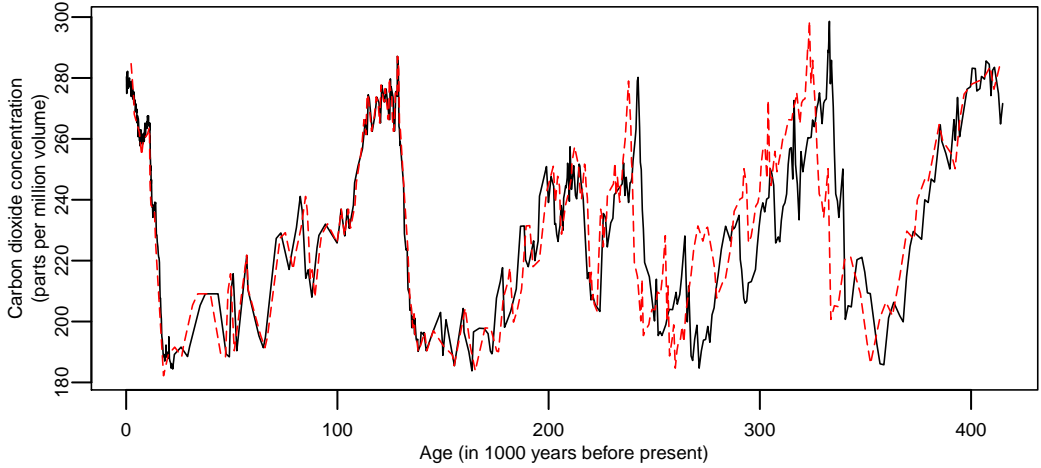}
            \caption{\label{fig:co2_epica_vostok}Atmospheric concentration of carbon dioxide derived from ice-cores at EPICA Dome C (solid line) and Lake Vostok (dashes) of Antarctica}
        \end{center}
    \end{figure}

    %As for pooling the data, there may be ambiguity about whether one should align the second data set to the first one, or vice versa.

    A similar problem in the literature of growth curves has been solved by what is known as {\it structural averaging}. Suppose there are $k$ sets of functional data $\{(t_{ij},y_{ij}); j=1\ldots n_i\},\;i=1\ldots k$ under the signal-plus-noise model,
    \begin{equation}
        y_{ij} = \mu(g_i(t_{ij}))+\epsilon_{ij},\quad j=1,\ldots,n_i,\ i=1,\ldots,k,\label{eq:other_model}
    \end{equation}
    where $\mu$ is the function of interest, $g_1,\ldots,g_k$ are the time-warping functions for different individuals $i=1,\ldots,k$, and $\epsilon_{ij}$'s are additive errors. In the structural averaging approach, the obvious non-identifiability of the functions $\mu$ and $g_1,\ldots,g_k$ is resolved by imposing a symmetric and additive constraint such as the average of the $g_i$'s \citep{Wang_Gasser_1997}, or the average of the $g_i^{-1}$'s \citep{Gervini_Gasser_2004}, is the identity map. Once these functions are estimated under the symmetric constraint, they are used to bring the data to a common time-scale, so that the `central' function $\mu$ can be estimated from the pooled data. Functional data on many individuals is expected to produce a good estimator of the `central' function.

    In growth models, one typically allows individual-specific multipliers to the function $\mu$ \citep{Gervini_Gasser_2004}, or even more general variations in $\mu$ \citep{Lawton_Sylvestre_Maggio_1972, Kneip_Gasser_1988, Kneip_Engel_1995}, but that would be unnecessary for the paleoclimatic problem mentioned above. For the latter problem, replication is also difficult to obtain -- not only for economic reasons, but also because of provisions of international treaties that prohibit intrusion in an ecologically sensitive area \citep{Watts_1992}. Thus, one cannot expect $k$ to be more than two or three.

    In particular, when $k=2$, a neutral time-scale for expressing a pooled estimate of the common function may be unnecessary. Instead, one can use the model
    \begin{eqnarray}
        y_{1j}&=&m(t_j)+\epsilon_{1j}\quad j=1,\ldots,n_1,\notag\\
        y_{2j}&=&m(g_0(s_j))+\epsilon_{2j}\quad j=1,\ldots,n_2,\label{eq:the_model}
    \end{eqnarray}
    where $m$ is the underlying mean function, expressed in the time-scale of the first data set, and $g_0$ is the time transformation function that warps the time-scale of the second data set into that of the first. The task of estimating the function $m$ from the above model may appear to be a special case of an already solved problem, since \eqref{eq:the_model} is apparently a simplified form of \eqref{eq:other_model}, under the asymmetric constraint $g_1(t)=t$. Moreover, this looks like a poorer formulation of the registration problem, as there is no reason why the time-scale of one data set should be preferred over that of the other one. However, we show in Section~\ref{sec:constraint_of_symmetry} that for $k=2$, the model~\eqref{eq:other_model} under a symmetric and additive constraint is more restrictive than the model \eqref{eq:the_model}. In particular, \eqref{eq:other_model} implies \eqref{eq:the_model} with $m=\mu\circ g_1$ and $g_0=g_1^{-1}\circ g_2$, but \eqref{eq:the_model} does not imply the existence of a pair of functions $g_1$ and $g_2$, satisfying a symmetric and additive constraint, so that \eqref{eq:other_model} may hold.

    In view of this annoying result, use of structural averaging in the present case (i.e., when $k=2$) is seen to be laden with an unnecessary restriction. If this restriction has to be avoided, the only available choice seems to be to warp the time-scales of any one of the data sets for aligning it with the other. Once this choice is made, the common function may be estimated through kernel smoothing of the pooled data, after substituting the function $g_0$ with a consistent estimator. We show in this paper that, under appropriate conditions, the resulting estimator of $m$ would be consistent.

\section{Limitations of the Constraint of Symmetry}\label{sec:constraint_of_symmetry}

    Suppose the model~\eqref{eq:other_model} for $k=2$ holds for some location function $\mu$ and strictly increasing time-warping functions $g_1$ and $g_2$ such that $\frac12(g_1(t)+g_2(t))=t$. This model can be easily expressed as the model~\eqref{eq:the_model} with $m(t)=\mu\comp g_1(t)$ and $g_0(t)=g_1^{-1}\comp g_2(t)$. The same conclusion holds when the constraint is $\frac12(g_1^{-1}(t)+g_2^{-1}(t))=t$.

    Now suppose we have two sets of functional data satisfying the model~\eqref{eq:the_model}, for some continuous and strictly increasing function $g_0$. We have to determine whether the data sets can also be said to follow the model~\eqref{eq:other_model} for some functions $g_1$ and $g_2$ such that $\frac12(g_1(t)+g_2(t))=t$, and if so, whether they are unique. We start by showing that if there is a suitable pair of warping functions $g_1$ and $g_2$, there can be no other choice.

    \begin{theorem}  \label{thm:uniqueness}
        Suppose $g_0$ is a continuous and strictly increasing function with domain $[c,d]\;(for\; c < d)$ and range $[a,b]\;(for\; a < b)$ such that $[a,b]\cap[c,d]$ includes a non-empty open interval, and $g_0(t)=t$ for some $t$ in $[a, b] \cap [c, d]$. If there exists an interval $[\alpha,\beta]\subseteq[a, b] \cap [c, d]$ and a pair of continuous and strictly increasing functions $g_1:\,[a,b]\rightarrow[\alpha,\beta]$ and $g_2:\,[c,d]\rightarrow[\alpha,\beta]$ such that for all $t\in [a, b] \cap [c, d]$,
        \begin{enumerate}
            \item[(a)] $g_0(t)=g_1^{-1}\circ g_2(t)$,
            \item[(b)] $\frac12(g_1(t)+g_2(t))=t$
        \end{enumerate}
          then the pair of functions $g_1$ and $g_2$ is unique.
    \end{theorem}

    The next theorem shows that, for a given $g_0$, existence of a pair of functions $g_1$ and $g_2$ satisfying conditions (a) and (b) of Theorem~\ref{thm:uniqueness} is not guaranteed.

    \begin{theorem}\label{thm:existence}
        Suppose
        \begin{displaymath}
            g_{0}(t) =
                \left\{
                    \begin{array}{rcl}
                        \frac{1-r(1-t_0)}{t_0}t, &  & t\in [0, t_0) \\
                        1-r(1-t), &  & t\in [t_0, 1].
                    \end{array}
                \right.\notag
        \end{displaymath}
        Then there is no pair of continuous and strictly increasing functions $g_1:\,[0,1]\rightarrow[0,1]$ and $g_2:\,[0,1]\rightarrow[0,1]$ such that the conditions (a)~$g_0(t)=g_1^{-1} \circ g_2(t)$ and (b)~$\frac12(g_1(t)+g_2(t))=t$ hold for all $t$ in $[0, 1]$.
    \end{theorem}

    The preceding theorem can be readily generalized to a piecewise linear map $g_0$ from domain $[c, d]$ to $[c, d]$, for any $c$ and $d$ with $c<d$. Further, for any $g_0(t)$ that intersects $t$ a few times, any candidate pair of functions $g_1$ and $g_2$ have to pass through these points of intersection. Therefore, the issue of existence of $g_1$ and $g_2$ has to be addressed separately for each interval in between consecutive crossings. If $g_0$ is piecewise linear, consisting of two lines in between any pair consecutive crossings, Theorem~\ref{thm:existence} implies that there can be no suitable $g_1$ and $g_2$ for that segment.

    Thus, existence of appropriate $g_1$ and $g_2$ is ruled out for a vast class of functions $g_0$. In other words, Theorem~\ref{thm:existence} indicates a general weakness of the model~\eqref{eq:other_model} under the symmetric constraint $\frac12(g_1(t)+g_2(t))=t$, rather than existence of a pathological counterexample.

    The model~\eqref{eq:other_model} has the same weakness under the alternative symmetric constraint $\frac12(g_1^{-1}(t)+g_2^{-1}(t))=t$. This is because Theorem~\ref{thm:existence} can be easily adapted to the replacement of condition~(b) by this constraint.

    There is a peculiar consequence of this limitation to inference. If the model~\eqref{eq:the_model} holds with some $g_0$ that happens to be piecewise linear (with two pieces) in between two crossings of $g_0(t)$ with $t$, then it would be impossible to get any pair of $g_1$ and $g_2$ under model~\eqref{eq:other_model} (with a symmetric and additive constraint) that would be commensurate with that $g_0$. In such a case, if one incorrectly assumes model~\eqref{eq:other_model} but `estimates' $g_1$ and $g_2$ under a symmetric and additive constraint, then the implied estimator of $g_0$, obtained by substitution of these estimators in $g_1^{-1}\comp g_2$, would have an unduly restrictive form. Since the true function $g_0$ does not have this form, the estimator would be biased and inconsistent.

    In summary, insistence on a neutral time scale for expressing the main function of interest can bring in an unnecessary constraint on the underlying warping function that relates the time scale of one data set to that of the other. It can also lead to avoidable bias in an estimator of this function. Given these limitations of the model~\eqref{eq:other_model}, it makes sense for us to bypass it and instead focus on the model~\eqref{eq:the_model}.

    \section{Estimation of Mean Function}\label{sec:est_of_mean_fn}

   Let us define, for a given time transformation~$g$, the real-valued function $m_{n,g}$ on the real line by
    \begin{equation}\label{eq:nw_type_functional}
        \displaystyle m_{n,g}(t)=
        \frac
        {
        \displaystyle\frac{1}{nh}\left\{\displaystyle\sum_{i=1}^{n_1}K\left(\frac{t-t_i}{h_n}\right) y_{1i}+\displaystyle\sum_{j=1}^{n_2}K\left(\frac{t-g(s_j)}{h_n}\right)y_{2j}\right\}
        }
        {
        \displaystyle\frac{1}{nh}\left\{\displaystyle \sum_{i=1}^{n_1}K\left(\frac{t-t_i}{h_n}\right) +\displaystyle\sum_{j=1}^{n_2}K\left(\frac{t-g(s_j)}{h_n}\right)\right\}
        },
    \end{equation}
    where $K$ is a kernel function, $h_n$ is the kernel bandwidth and $n=n_1+n_2$. Note that $m_{n,g_0}(t)$ is in fact the Nadaraya-Watson regression estimator of $m(t)$ based on the pooled data set, where time-values of the second data set are correctly re-warped by the transformation~$g_0$. As $g_0$ is not known, a natural estimator of $m(t)$ would be
    \begin{equation} \label{eq:nw_type_estimator}
        m_{n,\hat{g}_n}(t)=m_{n,g}(t)\big|_{g=\hat{g}_n},
    \end{equation}
    where $\hat{g}_n$ is an estimator of $g_0$.
    We establish in Section~\ref{sec:large_sample} that if $\hat{g}_n$ is consistent, then $m_{n,\hat{g}_n}$ would be consistent too, under appropriate conditions. We use the hypothetical estimator~$m_{n,g_0}$ to investigate the extent to which the performance of the Nadaraya-Watson estimator based on the first data set can be improved by making use of the second data set. In Section~\ref{sec:practical}, we suggest ways of computing the standard error of $m_{n,\hat{g}_n}$ and checking whether the second data set improves estimation. We study the performance of $m_{n,\hat{g}_n}$ through simulations in Section~\ref{sec:sim}, with $m_{n,g_0}$ and the Nadaraya-Watson estimator based on the first data set used as benchmarks. We analyse three sets of paleoclimatic data in Section~\ref{sec:data_anal} and provide some concluding remarks in Section~\ref{sec:discussion}.

\section{Large sample results}\label{sec:large_sample}

\subsection{Consistency of the proposed estimator}\label{subsec:con_res}

    Let the errors $\{\epsilon_{1i};i=1,\ldots n_1\}$ and $\{\epsilon_{2j};j=1,\ldots n_2\}$ and the time points $\{t_i;i=1,\ldots n_1\}$ and $\{s_j;j=1,\ldots n_2\}$ in model~\eqref{eq:the_model} be mutually independent sets of samples from the probability density functions $f_{\epsilon_1}$, $f_{\epsilon_2}$, $f_1$ and $f_2$ having supports over $(-\infty, \infty)$, $(-\infty, \infty)$, $[a,b]$ (for $a<b$) and $[c,d]$ (for $c<d$), respectively. Let $f_{\epsilon_1}$ and $f_{\epsilon_2}$ have mean 0 and variances $\sigma_{\epsilon_1}^2$ and $\sigma_{\epsilon_2}^2$, respectively.

    Suppose we have two sets of data following the model~\eqref{eq:the_model}. It is possible that some points in the second data set are mapped by $g_0$ beyond the interval $[a,b]$. Since the first data set contains no information in this region, there is no basis for guessing the function $g_0$ there. In other words, $g_0$ can be meaningfully estimated only where it takes values in the interval $[a,b]$. Therefore, we would seek to establish consistency of the estimator $m_{n,\hat{g}_n}$ only over the interior of the interval $[a,b]$.

    We can treat the right hand side of \eqref{eq:nw_type_functional} as a functional of $g$ for given $t$, and denote it by $M_{n,t}(g)$. Our first Theorem states the point-wise convergence of the functional $M_{n,t}$ on a suitable class of functions~$\G$.

    \begin{theorem}\label{thm:pointwise_conv}
        Suppose the following assumptions hold in respect of the model~\eqref{eq:the_model} and the functional $M_{n,t}$ defined above.
        \begin{description}
            \item [A1] The densities $f_1$ and $f_2$ are continuous, bounded and positive valued over the interior of their support.
            \item [A2] The mean function~$m:\R\rightarrow\R$ is continuous and bounded.
            \item [A3] The map $g_0:\R\rightarrow\R$ is strictly increasing and continuously differentiable.
            \item [A4] The kernel $K$ is a continuous and bounded probability density function, which is symmetric about zero.
            \item [A5] $h_n\rightarrow 0$, $nh_n\rightarrow \infty$, and $n_1/n\rightarrow~\xi$ as $n\rightarrow~\infty$ where $\xi\in(0,1)$.
        \end{description}
         Let $\G$ be the class of all strictly increasing and continuously differentiable functions~$g$ defined over~$\R$. Then, for any function $g\in\G$ and $t\in(a,b)$, $M_{n,t}(g)\conv{P}M_t(g)$ as $n\rightarrow\infty$, where
        \begin{eqnarray}\label{eq:M}
            M_t(g)=\frac{\xi m(t)f_1(t)+(1-\xi)m\circ g_0\circ g^{-1}(t)f_2\circ g^{-1}(t)(g^{-1})'(t)}{f_g(t)},
        \end{eqnarray}
        and $f_g(t)= \xi f_1(t)+(1-\xi)f_2\circ g^{-1}(t)(g^{-1})'(t)$.
    \end{theorem}

    Since $M_t(g_0)=m(t)$, it looks plausible that the above theorem would lead to the consistency of a plug-in estimator obtained from $M_{n,t}(g)$. Before we get there, we need to establish the uniform convergence of $M_{n,t}$, which has to happen within a compact subset of $\G$. Let us define the metric $\Delta(g_1,g_2)=\sup_{t}|g_1(t)-g_2(t)|$, for $g_1,g_2\in\G$.
    \begin{theorem}\label{thm:uniform_conv}
        Let $\G_0$ be a compact subset of\ \ $\G$ (defined in Theorem~\ref{thm:pointwise_conv}) in the metric space $(\G, \Delta)$ such that it includes $g_0$ and inverse functions of its members have second order derivative uniformly bounded across $\G_0$. Then, under Assumptions~A1, A2, A3, A5 and the additional assumption
        \begin{description}
            \item [A4*] The kernel $K$ in \eqref{eq:nw_type_functional} is a probability density function, which is symmetric about zero, continuous, bounded, bounded away from zero on a given closed interval, and has bounded first order derivative,
        \end{description}
        we have, for any $t\in(a,b)$, $$\sup_{g\in\G}|M_{n,t}(g)-M_t(g)|\conv{P}0\quad as\ n\rightarrow\infty.$$
    \end{theorem}
    Assumption~A4* is satisfied by e.g., Gaussian Kernel.

    Finally, we establish the consistency of a plug-in estimator in the next theorem.
    \begin{theorem}\label{thm:consistency}
        Suppose $\hat{g}_n$ is an estimator of $g_0$, which belongs to a compact set $\G_0$ as stated in Theorem~\ref{thm:uniform_conv}. Then, under the Assumptions~A1, A2, A3, A4*, A5 and appropriate additional conditions for ensuring consistency of $\hat{g}_n$, the estimator $m_{n,\hat{g}_n}(t)$ defined in~\eqref{eq:nw_type_estimator} converges in probability to $m(t)$ as $n\rightarrow\infty$ for all $t\in(a,b)$.
    \end{theorem}

    In the literature of curve registration, there are several estimators of $g_0$ that are consistent under various conditions. The additional requirement of Theorem~\ref{thm:consistency}, namely that~$\hat{g}_n$ should belong to~$\G_0$, has to be verified separately for each estimator~$\hat{g}_n$ and each compact set~$\G_0$. The Kernel-matched registration method \citep{Bhaumik_Srivastava_Sengupta_2017} is particularly suited to this device, as it produces a consistent estimator~$\hat{g}_n$ chosen within any given compact subset of $\G$. As long as the conditions for consistency of the estimator hold, no further verification is necessary.

    \subsection{Potential improvement in performance due to additional data}\label{subsec:improvement}

    As far as estimation of $m$ in \eqref{eq:the_model} is concerned, the idea of pooling the second data set with the first one stems from the expectation that more data would naturally lead to a better estimator. Let us examine if this really happens.

    The bias and the variance of the proposed plug-in estimator $m_{n,\hat{g}_n}(t)$ in \eqref{eq:nw_type_estimator} would depend on the choice of the plug-in estimator $\hat{g}_n$ of $g_0$. Instead of analyzing the performance of any particular estimator, we can consider the hypothetical estimator $m_{n,g_0}(t)$, which would be the appropriate choice if $g_0$ had been known. If this estimator performs better than the Nadaraya-Watson estimator based on the first data set alone, then that would indicate the  potential for improvement from the additional data. The issue is non-trivial, as it is not merely a matter of larger sample size. Even when $g_0$ is known, the model~\eqref{eq:the_model} has possibly different distributions of measurement errors and sampling times for the two sets of data, which sets it apart from the usual set-up used for analyzing the Nadaraya-Watson estimator.

    The next theorem gives the expression of mean squared error (MSE) of $m_{n,g_0}$.
    \begin{theorem}\label{thm:mse_g0}
        Suppose the following assumptions hold in respect of model~\eqref{eq:the_model} and the function~\eqref{eq:nw_type_functional}.
        \begin{description}
            \item[A1$'$] The densities $f_1$ and $f_2$ in model~\eqref{eq:the_model} have continuous second order derivatives and supports of $f_{\epsilon_1}$ and $f_{\epsilon_2}$ are bounded.
            \item[A2$'$] The mean function~$m:\R\rightarrow\R$ is bounded with continuous second order derivative.
            \item[A3$'$] The map $g_0:\R\rightarrow\R$ is strictly increasing and has continuous third order derivative.
            \item[A4$'$] Kernel $K$ is a continuous and compactly supported probability density function that is symmetric about zero.
            \item[A5$'$] $h_n\rightarrow 0$, $nh_n\rightarrow\infty$, as $n\rightarrow\infty$ and $n_1/n = \xi+o(h)$ where $\xi\in(0, 1)$.
 %           \item[A4$'$] $\max(|Y|,|Y'|)<C<\infty$ almost surely, \textcolor{blue}{assume boundedness of signal and error instead}
        \end{description}
        Then for every $t\in(a,b)$,
        \begin{eqnarray}\label{eq:thm:mse_g0:main}
            MSE(m_{n,g_0}(t))&\!\!=&\!\!\frac{\xi f_1(t)\sigma_{\epsilon_1}^2+ (1-\xi) f_2\circ g_0^{-1}(t)(g_0^{-1})'(t)\sigma_{\epsilon_2}^2}{nh_n}\times\frac{\|K\|_2^2}{f_{g_0}^2(t)}\notag\\
            &&\!\!+\frac{h_n^4}{4}\left[m''(t)+\frac{2m'(t)f_{g_0}'(t)}{f_{g_0}(t)}\right]^2\!\!\mu_2^2(K)+o\left(h_n^4+\frac{1}{nh_n}\right)\!,\qquad\mbox{}
        \end{eqnarray}
        where $f_{g_0}(t)= \xi f_1(t)+(1-\xi) f_2\circ g_0^{-1}(t)(g_0^{-1})'(t)$,
        $$ \|K\|_2^2=\int_{-\infty}^{\infty} K^2(s)ds,\quad and\quad \mu_2(K)=\int_{-\infty}^{\infty} s^2K(s)ds.$$
    \end{theorem}

    When $n_2=0$, the first term in the expression of the MSE simplifies to $\frac{1}{n_1h_{n_1}}\sigma_{\epsilon_1}^2\|K\|^2/f_1(t)$, which coincides with the leading term in the expression of variance of the Nadaraya-Watson estimator based only on the first sample \citep{Collomb1977a}. The ratio of this term (in the general case) to its value in the special case of $n_2=0$ simplifies to
    \begin{equation}
        \frac{n_1h_{n_1}}{\xi nh_n}\times\frac{1+ \frac{(1-\xi) f_2\circ g_0^{-1}(t)(g_0^{-1})'(t)}{\xi f_1(t)}\times\frac{\sigma_{\epsilon_2}^2}{\sigma_{\epsilon_1}^2}} {\left(1+ \frac{(1-\xi) f_2\circ g_0^{-1}(t)(g_0^{-1})'(t)}{\xi f_1(t)}\right)^2}.\label{eq:varratio}
    \end{equation}
    Under Assumption A5$'$, the limiting value of the first factor in~\eqref{eq:varratio} is the limiting value of $h_{n_1}/h_n$. In particular, if the bandwidth is chosen to be inversely proportional to one-fifth power of the sample size \citep{Hardle_1990}, this factor goes to~$\xi ^{-1/5}$. The value of the second factor depends on $t$. If $\sigma_{\epsilon_2}^2$ is less than $2\sigma_{\epsilon_1}^2$, this factor can be shown to have value in between~0 (when $\frac{\xi f_1(t)}{(1-\xi) f_2\circ g_0^{-1}(t)(g_0^{-1})'(t)}$ is close to~0) and~1 (when $\frac{(1-\xi) f_2\circ g_0^{-1}(t)(g_0^{-1})'(t)}{\xi f_1(t)}$ is close to~0). If $\sigma_{\epsilon_2}^2$ is greater than $2\sigma_{\epsilon_1}^2$, this factor can be larger than~1 (but less than $\frac{\sigma_{\epsilon_2}^4}{\sigma_{\epsilon_1}^2(\sigma_{\epsilon_2}^2-2\sigma_{\epsilon_1}^2)}$) in areas where $(1-\xi) f_2\circ g_0^{-1}(t)(g_0^{-1})'(t)$ is smaller than $(\frac{\sigma_{\epsilon_2}^2}{\sigma_{\epsilon_1}^2}-2)\xi f_1(t)$, i.e., where the second data set is likely to be sparse.

    The second term in the expression~\eqref{eq:thm:mse_g0:main} simplifies in the special case $n_2=0$ to
    $$\frac{h_{n_1}^4}{4}\left[m''(t)+\frac{2m'(t)f_1'(t)}{f_1(t)}\right]^2\mu_2^2(K),$$
    which coincides with the square of the leading term in the expression of bias of the Nadaraya-Watson estimator based only on the first sample \citep{Collomb1977a, Collomb1977b}. The ratio of the second term to its value in the special case $n_2=0$ is \begin{equation}
    \frac{h_n^4}{h_{n_1}^4}\times\frac{\left[m''(t)+2m'(t)\frac{d}{dt}\log f_{g_0}(t)\right]^2} {\left[m''(t)+2m'(t)\frac{d}{dt}\log f_1(t)\right]^2}.\label{eq:biasqratio}
    \end{equation}
    The density $f_{g_0}$ is a mixture of $f_1$ and a time-transformed version of $f_2$. If this mixture density happens to have less sharp peaks than that of $f_1$, the term added to $m''$ in the numerator of the second factor of~\eqref{eq:biasqratio} would be smaller in magnitude than the corresponding term in the denominator. In that case, the extremes of the numerator would be less pronounced than that of the denominator. The limiting value of the first factor, if the bandwidth is chosen as inversely proportional to one-fifth power of the sample size \citep{Hardle_1990} and Assumption~A5$'$\ holds, would be $\xi ^{4/5}$, which is less than 1. Thus, the first factor has a shrinking effect on the ratio~\eqref{eq:biasqratio}. The larger the sample size of the second data set, the smaller would be the presumed value of $\xi $, and greater the shrinking effect.

    The remainder term in the expression~\eqref{eq:thm:mse_g0:main} is of the same order as that of the Nadaraya-Watson estimator \citep{Schimek_2000}.

\section{Some practical issues\label{sec:practical}}

\subsection{Standard error of estimator\label{ssec:se_m_est}}

        The proposed estimator $m_{n,\hat{g}_n}(t)$ given in \eqref{eq:nw_type_estimator} is fully specified only after $\hat{g}_n$ is chosen. Therefore, we suggest here a model-based bootstrap approach for estimating the variance of $m_{n,\hat{g}_n}(t)$, which would work irrespective of the method of estimating~$g_0$. The basis of such a bootstrap scheme would be the model~\eqref{eq:the_model}, with the functions $g_0$ and $m$, as well as the underlying distributions replaced by suitable estimates. For estimation of $m$, one can use $m_{n,\hat{g}_n}$ where $\hat{g}_n$ is the consistent estimator of $g_0$. In order to avoid ties in the resampled times, one can draw samples from kernel density estimators of $f_1$ and $f_2$, and adopt the same strategy for drawing samples from estimated $f_{\epsilon_1}$ and $f_{\epsilon_2}$, by using the residuals in the two samples as proxies of the respective model errors. One can use these bootstrap estimates to obtain not only the pointwise standard errors, but also confidence intervals and confidence bands.

\subsection{Usefulness of second data set\label{ssec:cvcomp}}

    In order to judge whether the second data set has indeed led to an improved estimator of the mean function $m$ in the time scale of the first data set, one may use the average of leave-one-out cross-validation squared prediction errors. Specifically, in the case of the pooled data set, the warping function $g_0$ has to be re-estimated after the deletion of every single observation from either data set. If the value of the cross-validation criterion for the pooled data set is smaller than its value for the first data set, then it may be concluded that the second data set has been useful.

\section{Simulation of Performance}\label{sec:sim}

    In order to study the possible improvement in estimation of the mean function from the use of an additional dataset, we consider the estimators $m_{n,\hat{g}_n}(t)$ given in \eqref{eq:nw_type_estimator} and the Nadaraya-Watson regression estimator for the first data set (denoted here by $m_{nw}(t)$). As a benchmark we also compare these estimators with the hypothetical estimator $m_{n,g_0}(t)$ based on the correct transformation $g_0$. The version of the estimator $m_{n,\hat{g}_n}(t)$ used in this study is based on the kernel-matched regression estimator of the warping function
    \begin{equation}
    \hat{g}_n=\arg\max_{g\in\G_0}\left\{ \frac
                {
                \displaystyle\frac{1}{n_1 n_2}
                \displaystyle\sum\limits_{i=1}^{n_1}
                \displaystyle\sum\limits_{j=1}^{n_2}
                \frac{1}{h_t}K_1\left(\frac{{t_i} - g({s_j})}{h_t}\right)
                \frac{1}{h_y}K_2\left(\frac{{y_{1i}} - y_{2j}}{h_y}\right)
                }
                {
                \displaystyle\frac{1}{n_1 n_2}\displaystyle\sum\limits_{i=1}^{n_1}
                \displaystyle\sum\limits_{j=1}^{n_2}
                \frac{1}{h_t}K_1\left(\frac{{t_i} - g(s_j)}{h_t}\right)}\right\},\label{eq:ghatn}
    \end{equation}
    proposed by \citep{Bhaumik_Srivastava_Sengupta_2017}, where $\G_0$ is the vector space generated by linear B-spline basis functions with equidistant knot points.

    The mean function chosen for the simulation is shown in Figure~\ref{fig:sim_m_fun}. This function is similar to the apparent movement of the carbon dioxide data series analysed in the next section. The first data set was obtained from the first equation of model~\eqref{eq:the_model} with 500 time samples chosen from the uniform distribution over the time range $[0, 415]$ and the additive errors were normal with mean zero and standard deviation equal to 10 per cent of the the empirical standard deviation of the $m(t_j)$'s. The time transformation function used to define the time scale of the second data set was:
    \begin{equation}
        g_0(t)=t+0.05t\sin\left(\frac{4\pi t}{415}\right),\quad 0\le t\le415.
    \end{equation}
    The density of the time samples of the second data set was chosen as,
    $$
        f_2(t)=\frac{t/415+\frac12}{415},\quad 0\le t\le415,
    $$
    and the distribution of the additive error was chosen as in the first data set. The size of the second data set was also 500. Note that the chosen $g_0$ is not a member of the search space of the estimator, described above.
   \begin{figure}[h!]
        \begin{center}
            \includegraphics{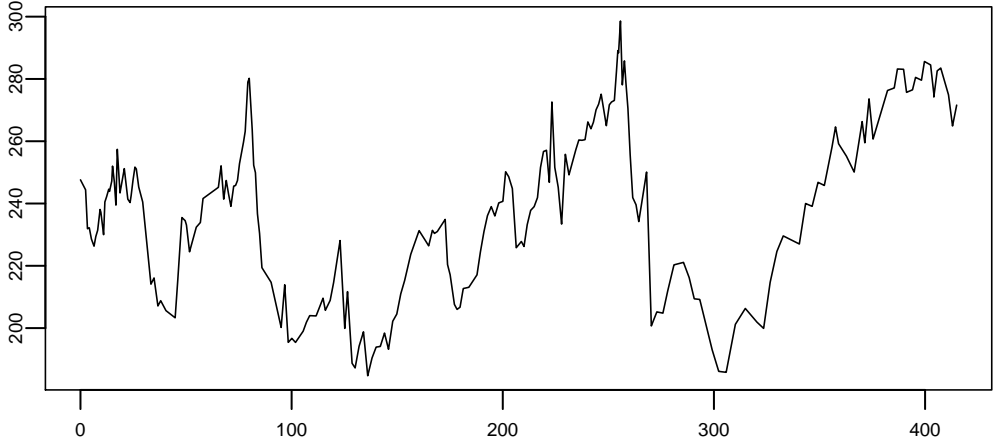}
            \caption{The mean function chosen for the simulation exercise}\label{fig:sim_m_fun}
        \end{center}
    \end{figure}

    Following the prescription given in~\citet{Bhaumik_Srivastava_Sengupta_2017}, a total of 30 equidistant knot points were chosen for estimation of $g_0$, Gaussian kernels were used in \eqref{eq:ghatn}, $h_t$ was chosen as half of the average horizontal separation between successive observations of the lesser dense data set and $h_y$ was chosen as 10 per cent of the combined range of the $y$-values of the two data sets. Following the computations in~\citet{Bhaumik_Srivastava_Sengupta_2017}, the optimization in \eqref{eq:ghatn} was done through multiple rounds of sequential grid search of the coefficients, arranged in their natural order (from left to right). The iterations for $g$ were started with the identity map. The Gaussian kernel truncated at $\pm8$ was used in \eqref{eq:nw_type_estimator}, while its bandwidth $h_n$ was chosen by the method of leave-one-out cross validation.

    \begin{figure}[h!]
        \begin{center}
            \includegraphics{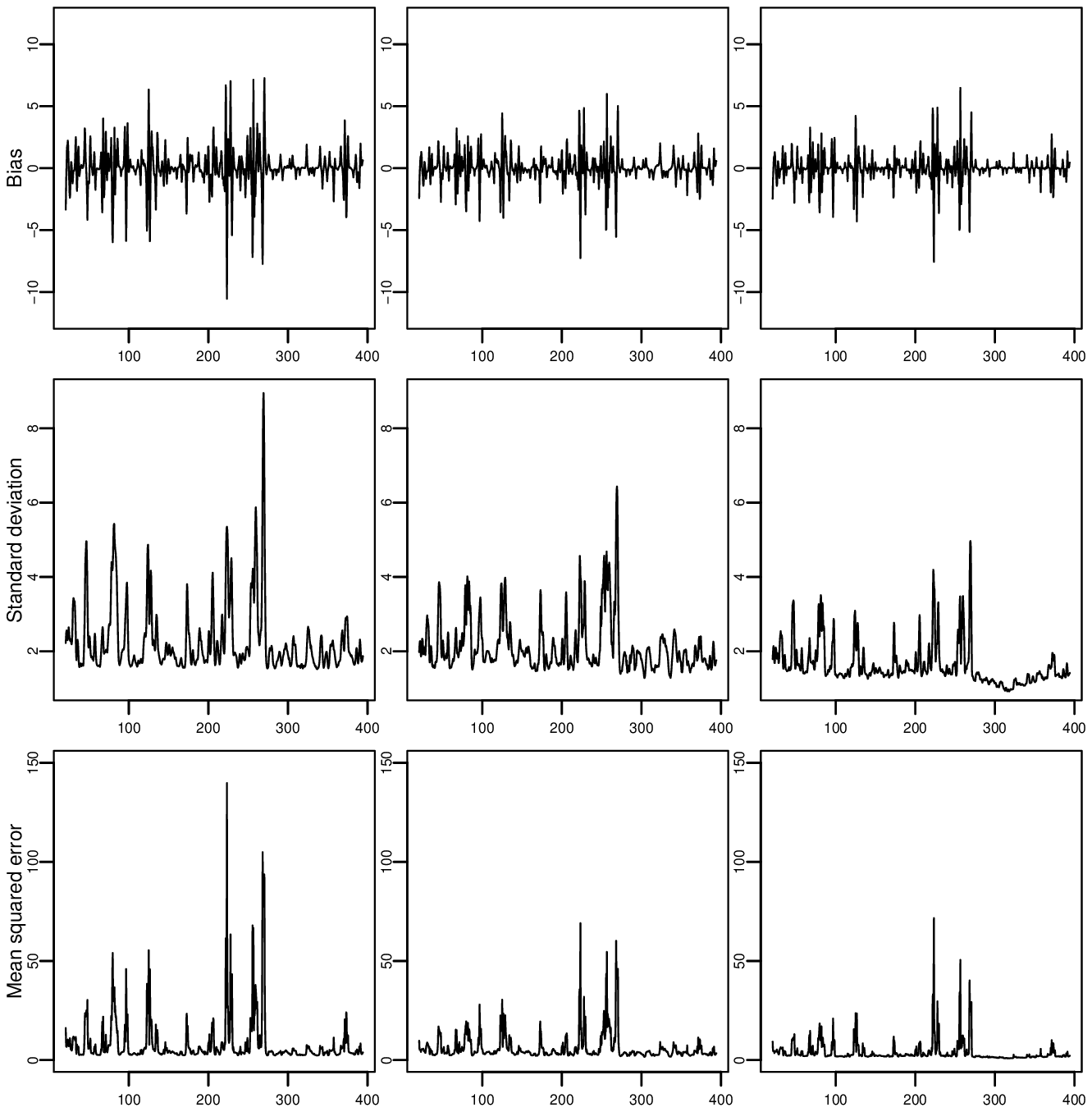}
            \caption{Simulated point-wise bias (first row), standard deviation (second row) and mean squared error (last row) of the estimates of mean function  by $m_{nw}(t)$ (first column), $m_{n,\hat{g}_n}(t)$ (second column) and $m_{n,g_0}(t)$ (last column)}\label{ch4:fig:bias_sd_mse}
        \end{center}
    \end{figure}

      \begin{figure}[h!]
        \begin{center}
            \includegraphics{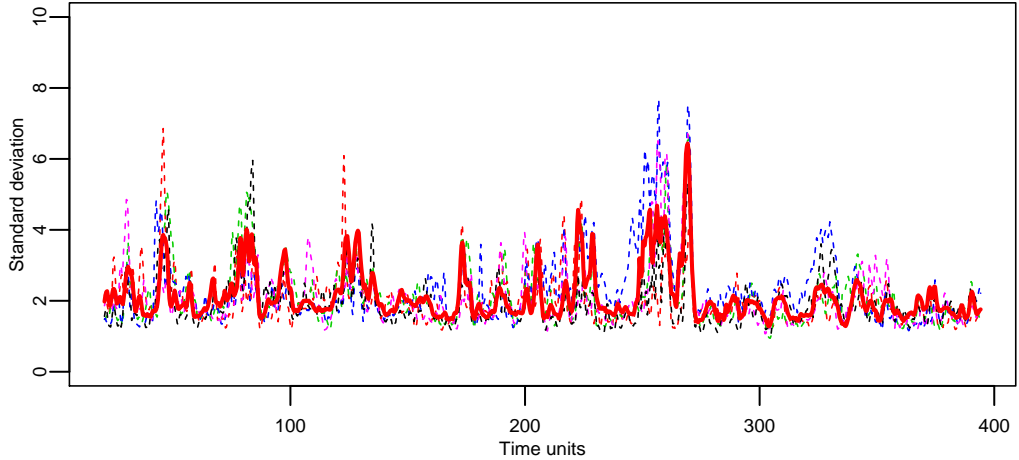}
            \caption{Bootstrap estimates of standard deviation of the estimated mean function by kernel-matched registration from five different simulation runs (dashes), with the empirical standard deviation obtained from 1000 simulation runs (thick solid line)}\label{ch4:fig:btstrp_se_of_m_est}
        \end{center}
    \end{figure}

    \begin{table}
        \caption{\label{tab:des_stat}Some descriptive statistics of the data sets}
        \centering
        \begin{tabular}{lrcrc}
             \hline
             & \multicolumn{2}{c}{Data set 1: EPICA Dome C} & \multicolumn{2}{c}{Data set 2: Vostok} \\
             Data&Size&Range ($Y$-Value)&Size&Range ($Y$-Value)\\
             \hline
             Carbon dioxide&537&183.8-298.6&283&182.2-298.7\\
             Methane&1,545&342-907&457&318-773\\
             Temp. deviations&5,028&(-)10.58-5.46&3,310&(-)9.39-3.23\\
             \hline
        \end{tabular}
    \end{table}
    The average pointwise bias, standard deviation and mean squared error of the three estimators, based on 1000 simulation runs, are shown in Figure~\ref{ch4:fig:bias_sd_mse}. The variance of the hypothetical estimator $m_{n,g_0}(t)$ is smaller towards the right side of the time axis. This pattern is anticipated from the discussion of section~\ref{subsec:improvement}, as the second data set has higher data density towards the right. The same pattern is seen in the pointwise standard deviation of~$m_{n,\hat{g}_n}(t)$. Overall, the proposed estimator~$m_{n,\hat{g}_n}(t)$ has better performance than~$m_{nw}(t)$ in all respects. A summary of the performances is captured by the integrated mean squared errors of $m_{nw}(t)$, $m_{n,\hat{g}_n}(t)$ and $m_{n,g_0}(t)$, which turned out to be 8.2, 6.1 and 3.8, respectively. Thus, the proposed estimator is able to exploit information from the second data set for improved estimation, although there is some loss incurred for not knowing the true transformation function.

    Figure~\ref{ch4:fig:btstrp_se_of_m_est} shows the simulated standard deviation of the estimated mean function, computed from 1000 runs, together with five examples of bootstrap estimates of the same (for the first five runs), each computed from 1000 resamples. The bootstrap estimates appear to be in line with the simulated pointwise standard error.

\section{Analysis of ice core data}\label{sec:data_anal}
We now return to the pleoclimatic data mentioned at the beginning of the paper. The data, gathered from ice cores collected at EPICA Dome C and over Lake Vostok in Antarctica, consist of (i) carbon dioxide~\citep{Luthi_et_al_2008, Petit_et_al_1999}, (ii) methane~\citep{Loulergue_et_al_2008, Petit_et_al_1999} and (iii) average annual temperature deviation~\citep{Jouzel_et_al_2007, Petit_et_al_1999} at different times measured in thousand years before present (YBP). Table~\ref{tab:des_stat} gives the summary of these variables in the two data sets. We chose the data sets from EPICA Dome C and Lake Vostok as the first and the second data sets, respectively, in~\eqref{eq:the_model}.

This choice was based on the fact that the data density of all the three series is higher at EPICA Dome C. The parameters of the estimators were chosen as in previous section. The overlaid plots of the proposed estimate~$m_{n,\hat{g}_n}(t)$ and the Nadaraya-Watson estimate~$m_{nw}(t)$ of the mean functions of the three data sets are shown in Figure~\ref{fig:m_est_co2_ch4_temp}. For the carbon dioxide data set, the peaks and valleys of the pooled estimate are sharper. Thus, the combined data indicate a slightly wider range of historic variation in this variable than the range suggested by the EPICA Dome C data alone. There is not much difference in the range of values of the two estimates of the mean function of the methane data. In contrast, the pooled estimate has smaller range of variation in the case of temperature deviation data. The likely reason for this is the discrepancy between the two sets of temperature data.

%It may be noticed from Table~\ref{tab:des_stat} that the temperature series has higher data density than the other two series for both the locations. The combined information from the two data sets on temperature deviation indicates that this variable has had slower movements than what might have been thought from the EPICA Dome C data set alone. It may also be noticed from Figure~\ref{fig:m_est_co2_ch4_temp}(c) that the pooled estimate of temperature deviation is generally larger than the estimate based on the first data set during the period 20,000 to 100,000 YBP (in the EPICA Dome C time scale). Upon closer examination, the pooled estimate is found to be larger by about 0.9 degree Celsius during this phase.

In order to determine whether the Lake Vostok data, after being registered and pooled with the EPICA dome C data, improve the estimation of the mean function, we computed the average of leave-one-out cross-validation squared prediction errors, described in Section~\ref{ssec:cvcomp}. The computed values of this metric, with and without the Lake Vostok data, on the three variables are reported in Table~\ref{tab:cv_err}. It is found that the use of the Lake Vostok data appears to improve estimation in the case of carbon dioxide, but not in the cases of methane and temperature. This may have happened because of discrepancy between the unaligned data sets plotted in the top rows of Figures~7-9 of \citet{Bhaumik_Srivastava_Sengupta_2017}, which is higher in the cases of methane and temperature.

    \begin{figure}[h!]
        \begin{center}
            \includegraphics{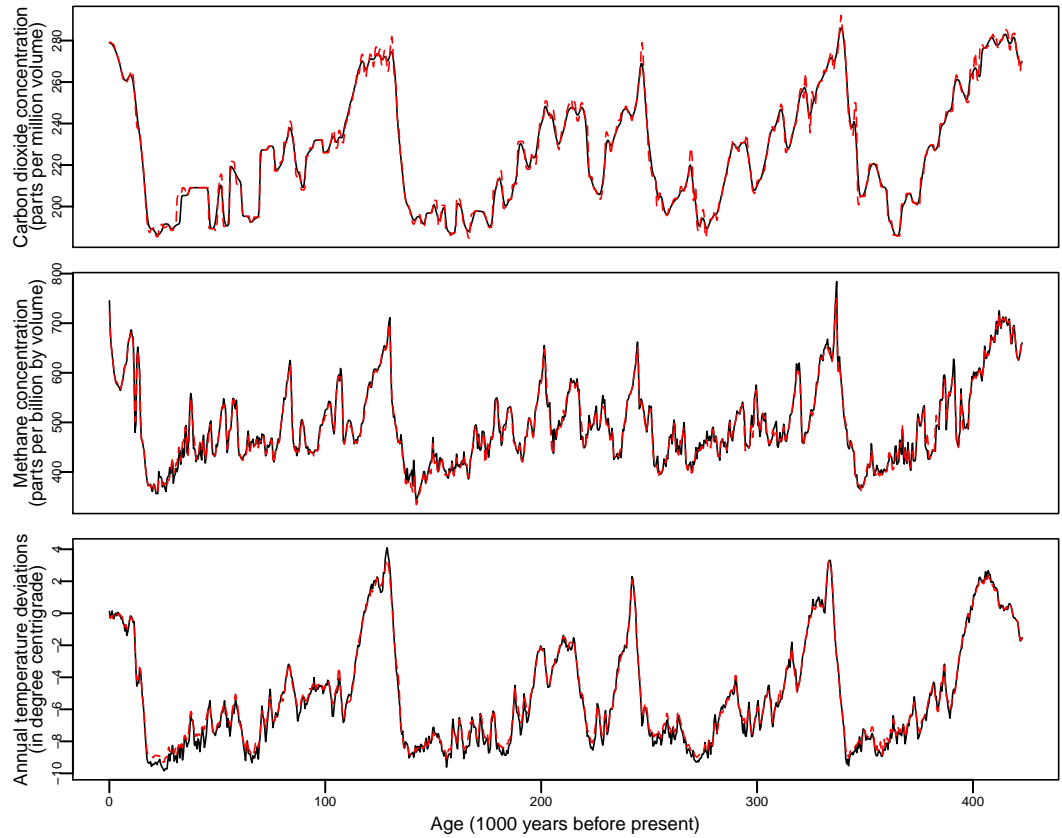}
            \caption{The overlaid plots of $m_{n,\hat{g}_n}(t)$ (dashes) and $m_{nw}(t)$ (dots) of the mean functions for the data sets on (a) carbon dioxide (top), (b) methane (middle) and (c) temperature deviations (bottom)}\label{fig:m_est_co2_ch4_temp}
        \end{center}
    \end{figure}

    \begin{table}
        \caption{\label{tab:cv_err}Average of leave-one-out cross-validation squared prediction errors of estimated mean functions of three variables}
        \centering
        \begin{tabular}{lrr}
             \hline
             &EPICA&EPICA Dome C\\
             &Dome C&data pooled\\
             &data only   &with aligned\\
             &&Lake Vostok data\\
             \hline
             Atmospheric concentration of carbon dioxide&38.2&18.1\\
             Atmospheric concentration of methane&341.9&571.9\\
             Average annual temperature deviation&0.3260&0.4374\\
             \hline
        \end{tabular}
    \end{table}

\section{Concluding remarks}\label{sec:discussion}
    We have shown in this paper how two sets of paleoclimatic functional data can be pooled, after registration, for consistent and improved estimation of the mean function. Analysis of the three pairs of data sets from EPICA Dome C and Lake Vostok has revealed some interesting departures resulting from the contributions of the smaller (Lake Vostok) data set.

    Even though the collection of paleoclimatic data from ice core experiments is very expensive, there are some shorter data sets on similar variables (see e.g. \citet{Blunier_et_al_1995}, \citet{Anklin_et_al_1997}, \citet{Smith_et_al_1997}, \citet{Brook_et_al_1997} and so on). Assuming that there would ever be only a handful of such data sets, structural averaging is not an option for estimation of the mean function, for reasons explained in Section~\ref{sec:constraint_of_symmetry}. One can register all the data sets with respect to one of them, possibly the one having the largest data density. The theorems presented in Section~\ref{subsec:con_res} should be extendable to more than two data sets, eventually leading to the consistency of a pooled estimator.

\appendix
\section{Appendix: Proofs of theoretical results}

        %%%%%%%%%%%%%% Proof of Theorem 2.1: Uniqueness %%%%%%%%%%%%%%%%%%%

    \begin{proof}[Proof of Theorem~\ref{thm:uniqueness}.]
        For contradiction, let $g_1^*$ and $g_2^*$ be an alternative pair of functions sharing all the properties of $g_1$ and $g_2$. The function $f=g_1-g_1^*$ satisfies the identity
        \begin{eqnarray}
        f\circ g_0(t)+f(t)&=& g_1\circ g_0(t)-g_1^*\circ g_0(t)+g_1(t)-g_1^*(t)\nonumber\\
        &=&[2t-g_1(t)]-[2t-g_1^*(t)]+g_1(t)-g_1^*(t)\nonumber\\
        &=&0\label{fsum}
        \end{eqnarray}
        for $t,g_0(t)$ lying in the interval $[a, b] \cap [c, d]$.

        Let $u_0\in [a, b] \cap [c, d]$ be such that $g_1(u_0)\ne g_1^*(u_0)$. Note that $g_0(u_0)$ has to be different from $u_0$ (otherwise $g_1(u_0)$, $g_2(u_0)$, $g_1^*(u_0)$ and $g_2^*(u_0)$ would all coincide with $u_0$). The monotonicity of $g_0$ implies that the graph of $g_0(t)$ does not intersect that of $t$ anywhere in the points $u_0$ and $g_0(u_0)$ on the real line, whichever happens to be to the right of the other. Suppose
        \begin{eqnarray*}
        S_l&=&\{t\in[a,b]\cap[c,d]:\ t<g_0(u_0),\ g_0(t)=t\},\\
        S_u&=&\{t\in[a,b]\cap[c,d]:\ t>u_0,\ g_0(t)=t\}.
        \end{eqnarray*}
        At least one of the sets $S_l$ and $S_u$ is non-empty, because their union includes at least one point, where the graph of $g_0(t)$ intersects that of $t$ (according to the form of $g_0$ hypothesized in the theorem). As argued above, $g_1$ and $g_1^*$ must coincide over $S_l\cup S_u$.

        Case I: $S_l$ is non-empty. Define $t_l=\sup S_l$. If $g_0(u_0)<u_0$, then $u_0$, $g_0(u_0)$, $g_0\circ g(u_0),\ldots$ is a decreasing sequence bounded from below by $t_l$, and should converge to $t_l$. If $g_0(u_0)>u_0$, then $u_0$, $g_0^{-1}(u_0)$, $g_0^{-1}(g_0^{-1}(u_0))$, $\ldots$ is a decreasing sequence bounded from below by $t_l$, and should converge to $t_l$. In either case, let us denote the convergent sequence by $u_0$, $u_1$, $u_2$, $\ldots$ and note from \eqref{fsum} that $f(u_i)=(-1)^if(u_0)$. However, by continuity of $f$ (which follows from the continuity of $g_1$ and $g_1^*$), this alternating sequence should converge to $f(t_l)$, which is equal to~0. Therefore,\linebreak $f(u_0)=0$.

        Case II: $S_l$ is empty, but $S_u$ is non-empty. Define $t_u=\inf S_u$. If $g_0(u_0)>u_0$, then $u_0$, $g_0(u_0)$, $g_0\circ g_0(u_0)$, $\ldots$ is an increasing sequence bounded from above by $t_u$, and should converge to $t_u$. If $g_0(u_0)<u_0$, then $u_0$, $g_0^{-1}(u_0)$, $g_0^{-1}\circ g_0^{-1}(u_0)$, $\ldots$ is an increasing sequence bounded from above by $t_u$, and should converge to $t_u$. In either case, let us denote the convergent sequence by $u_0$, $u_1$, $u_2$, $\ldots$ and note from \eqref{fsum} that $f(u_i)=(-1)^if(u_0)$. However, by continuity of $f$, this alternating sequence should converge to $f(t_u)$, which is equal to~0. Therefore, $f(u_0)=0.$
    \end{proof}

        %---------- End of Proof of Theorem 2.1: Uniqueness ----------%

        \medskip

        %%%%%%%%%%%%%% Proof of Theorem 2.2: Non-existence %%%%%%%%%%%%%%%%%%%

        \begin{proof}[Proof of Theorem~\ref{thm:existence}.]
            Define $g_1,g_2:\,[t_0,1]\rightarrow[0,1]$ as
            \begin{eqnarray*}
                g_1(t)&=&1-\frac{2(1-t)}{1+r},\\
                g_2(t)&=&1-\frac{2r(1-t)}{1+r}.
            \end{eqnarray*}
            It is easy to see that the requisite identities (a) and (b) hold over $[t_0,1]$. By Theorem~\ref{thm:uniqueness}, this pair of functions would be the unique solution, had the problem been restricted to the interval $[t_0,1]$.

            Now suppose, for contradiction, $g_1$ and $g_2$ exist in the requisite form over the entire interval [0,1]. Since the requisite functional identities must hold over the sub-interval $[t_0,1]$ also, the requisite functions $g_1$ and $g_2$ must be extensions of the functions defined above.

            Suppose, for $n\ge 1$, $t_n=g_0^{-1}(t_{n-1})$, i.e., $t_n=v^nt_0$ with $v=[1-r(1-t_0)]/t_0$. Since $v<1$, this sequence of points lie in the interval $[0,t_0)$ and converges to 0 as $n\rightarrow\infty$. Even though we have not defined $g_1$ and $g_2$ over $[0,t_0)$, we can define them at these points by making use of the functional identities. In particular,
            \begin{eqnarray*}
                g_1(t_n)&=&2t_n-g_2(t_n)\ \ =\ \ 2t_n-g_1\circ g_0(t_n)\\
                &=&2t_n-g_1(t_{n-1})\ \ =\ \ 2t_n-2t_{n-1}+g_1(t_{n-2})\\
                &=&2t_n-2t_{n-1}+2t_{n-2}-\cdots+2(-1)^{n-1}t_1+(-1)^ng_1(t_0)\\
                &=&(-1)^ng_1(t_0)+2t_0[v^n-v^{n-1}+\cdots+(-1)^{n-1}v].
            \end{eqnarray*}
            Therefore, when $n$ is even, we have
            \begin{eqnarray*}
                g_1(t_n)-g_1(t_{n+1})&=&g_1(t_0)+2t_0[v^n-v^{n-1}+\cdots-v]\\
                &&+g_1(t_0)-2t_0[v^{n+1}-v^n+v^{n-1}-\cdots+v]\\
                &=&2g_1(t_0)-2v^{n+1}t_0+4t_0[(v^n-v^{n-1})+\cdots(v^2-v)]\\
                &=&2g_1(t_0)-2v^{n+1}t_0+4t_0(v-1)[v^{n-1}+v^{n-3}+\cdots+v]\\
                &=&2g_1(t_0)-2v^{n+1}t_0-\frac{4t_0(1-v)v(1-v^n)}{1-v^2}\\
                &=&2g_1(t_0)-2v^{n+1}t_0-\frac{4t_0v(1-v^n)}{1+v}\\
                &=&2g_1(t_0)-\frac{4t_0v}{1+v}+2t_0\left[\frac{1-v}{1+v}\right]v^{n+1}.
            \end{eqnarray*}
            When $n$ is odd, we have
            \begin{eqnarray*}
                g_1(t_n)-g_1(t_{n+1})&=&-g_1(t_0)+2t_0[v^n-v^{n-1}+\cdots+v]\\
                &&-g_1(t_0)-2t_0[v^{n+1}-v^n+v^{n-1}-\cdots-v]\\
                &=&-2g_1(t_0)+2v^{n+1}t_0-4t_0[(v^{n+1}-v^n)+\cdots(v^2-v)]\\
                &=&-2g_1(t_0)+2v^{n+1}t_0-4t_0(v-1)[v^n+v^{n-2}+\cdots+v]\\
                &=&-2g_1(t_0)+2v^{n+1}t_0+\frac{4t_0(1-v)v(1-v^{n+1})}{1-v^2}\\
                &=&-2g_1(t_0)+2v^{n+1}t_0+\frac{4t_0v(1-v^{n+1})}{1+v}\\
                &=&-2g_1(t_0)+\frac{4t_0v}{1+v}+2t_0\left[\frac{1-v}{1+v}\right]v^{n+1}.
            \end{eqnarray*}
            For $g_1$ to be continuous at 0, the above two series should converge to zero. Therefore, one must have $g_1(t_0)=\frac{2t_0v}{1+v}$, i.e.,
            $$1-\frac{2(1-t_0)}{1+r}=\frac{2t_0[1-r(1-t_0)]}{t_0+1-r(1-t_0)}.$$
            This constraint can be rewritten as a quadratic equation in $r$, which has the roots 1 and $\frac{1+2t_0}{1-2t_0}$. The root $r=1$ corresponds to the trivial case of $g_0(t)=t$ for all $t\in[0,1]$. The other choice leads to negative value of $r$ or of $g(t_0)$. Thus, no solution exists.
        \end{proof}

        %---------- End of Proof of Theorem 2.2: Non-existence ----------%

        \medskip

        %%%%%%%%%%%%%% Proof of Theorem 3.1: Point-wise convergence %%%%%%%%%%%%%%%%%%%

     \begin{proof}[Proof of Theorem~\ref{thm:pointwise_conv}.]
        Write $M_{n,t}(g)$ described in \S~\ref{subsec:con_res} as
            \begin{eqnarray}
                M_{n,t}(g)=\frac{N_{n,t}(g)}{D_{n,t}(g)}\label{eq:Mn2},
            \end{eqnarray}
            and $M_t(g)$ in~(\ref{eq:M}) as
            \begin{eqnarray}
                M_t(g)=\frac{N_t(g)}{D_t(g)}\label{eq:M2},
            \end{eqnarray}
            where
            \begin{eqnarray}
                N_{n,t}(g)&=&\frac{1}{nh_n}\left\{\displaystyle\sum_{i=1}^{n_1}K\left(\frac{t-t_i}{h_n}\right) y_{1i}+\displaystyle\sum_{j=1}^{n_2}K\left(\frac{t-g(s_j)}{h_n}\right)y_{2j}\right\},\label{eq:Nn}\\
                D_{n,t}(g)&=&\frac{1}{nh_n}\left\{\displaystyle\sum_{i=1}^{n_1}K\left(\frac{t-t_i}{h_n}\right) +\displaystyle\sum_{j=1}^{n_2}K\left(\frac{t-g(s_j)}{h_n}\right)\right\}\label{eq:Dn},\\
                N_t(g)&=&\displaystyle\lambda_1m(t)f_1(t)+\lambda_2m\circ g_0\circ g^{-1}(t)f_2\circ g^{-1}(t)(g^{-1})'(t)\label{eq:N},\\
                D_t(g)&=&\displaystyle\lambda_1f_1(t)+\lambda_2f_2\circ g^{-1}(t)(g^{-1})'(t)\label{eq:D}.
            \end{eqnarray}
            We aim to prove the stated result by showing that
            \begin{equation}
            N_{n,t}(g)\conv{P}N_t(g)\mbox{ and }D_{n,t}(g)\conv{P}D_t(g)\mbox{ as }n\rightarrow\infty,\label{eq:NnDn_conv}
            \end{equation}
            and by using the Continuous Mapping Theorem.

            We write the expected value of~\eqref{eq:Nn} as
            \begin{eqnarray}
                E(N_{n,t}(g)) = E_{1n}+E_{2n},\notag
            \end{eqnarray}
            where
            \begin{eqnarray}
                E_{1n} &=& \frac{n_1}{n}\frac{1}{h_n}\int_a^bK\left(\frac{x-s}{h_n}\right)m(s)f_1(s)ds, \notag\\
                E_{2n} &=& \frac{n_2}{n}\frac{1}{h_n}\int_c^dK\left(\frac{t-g(s)}{h_n}\right)m\circ g_0(s)f_2(s)ds\notag.
            \end{eqnarray}
            By substituting $u=\frac{t-s}{h_n}$, and defining an extended version of $f_1$ to be~0 outside its support, one can write $E_{1n}$ as
            $$
                \frac{n_1}{n}\int_{-\infty}^{\infty}I_{\left(-\frac{b-t}{h_n},\frac{t-a}{h_n}\right)}(u)K(u)m(t-h_nu)f_1(t-h_nu)du.
            $$
            By Assumptions~A1,~A2, and that $m$ is bounded, it follows that the integrand is bounded. Further by Assumptions~A1,~A2,~A3 and that $m$ is continuous, it follows that the integrand converges to $K(u)m(t)f_1(t)$ as $n$ tends to infinity. Hence, by the dominated convergence theorem and Assumption~A3, we have
            $$
                E_{1n}\rightarrow\lambda_1m(t)f_1(t)\int_{-\infty}^{\infty}K(u)du
            $$
            as $n$ tends to infinity. Since $K$ is assumed to be a pdf (Assumption~A2)
            $$
                E_{1n}\rightarrow\lambda_1m(t)f_1(t)\quad\mbox{as }n\rightarrow\infty.
            $$
            The substitution $\frac{t-g(s)}{h_n}=u$, and an appropriate extension of $f_2$ beyond its support, allows one to write the second part of $E(N_{n,t}(g))$ as
            \begin{eqnarray}
                E_{2n}&=&\frac{n_2}{n}\int_{-\infty}^{\infty} I_{\left(-\left(\frac{g(d)-t}{h_n}\right),\left(\frac{t-g(c)}{h_n}\right)\right)} (u)K(u)m\circ g_0\circ g^{-1}(t-h_nu)\notag\\
                &&\qquad\times f_2(g^{-1}(t-h_nu))(g^{-1})'(t-h_nu)du.\notag
            \end{eqnarray}
            Since $g\in\G$, it follows from Assumptions~A1--A3 that
            $$E_{2n}\rightarrow\lambda_2m(g_0(g^{-1}(t)))f_2(g^{-1}(t))(g^{-1})'(t),$$ which implies
            \begin{eqnarray}
                E(N_n(g))\rightarrow N(g)\quad\mbox{as } n\rightarrow\infty.\label{eq:ENn_conv}
            \end{eqnarray}

            In order to calculate the limiting variance of \eqref{eq:Nn}, we observe
            $$N_n^2(g) = T_{1n}+T_{2n}+T_{3n}+T_{4n}+T_{5n},$$
            where
            \begin{eqnarray*}
                T_{1n} &=& \frac{1}{n^2h_n^2} \sum_{i=1}^{n_1}K^2\left(\frac{t-t_i}{h_n}\right)y_{1i}^2,\\
                T_{2n} &=& \frac{1}{n^2h_n^2} \sum_{j=1}^{n_2}K^2\left(\frac{t-g(s_j)}{h_n}\right)y_{2j}^2,\\
                T_{3n} &=& \frac{1}{n^2h_n^2} \sum_{i=1}^{n_1}\sum_{\substack{i'=1 \\ (\neq i)}}^{n_1}K\left(\frac{t-t_i}{h_n}\right)y_{1i}K\left(\frac{t-t_{i'}}{h_n}\right)y_{1i'},\\
                T_{4n} &=& \frac{1}{n^2h_n^2} \sum_{j=1}^{n_2}\sum_{\substack{j'=1 \\ (\neq j)}}^{n_2}K\left(\frac{t-g(s_j)}{h_n}\right)y_{2j}K\left(\frac{t-g(s_{j'})}{h_n}\right)y_{2j'},\\
                T_{5n} &=& \frac{2}{n^2h_n^2} \sum_{i=1}^{n_1}\sum_{j=1}^{n_2}K\left(\frac{t-t_i}{h_n}\right)y_{1i}K\left(\frac{t-g(s_{j})}{h_n}\right)y_{2j}.
            \end{eqnarray*}
            The expected value of the first term is
            \begin{eqnarray*}
               &&\hskip-25pt  E(T_{1n})\\&=&\frac{n_1}{n^2h_n^2}\int_a^bK^2\left(\frac{t-s}{h_n}\right)E(Y_s^2|s)f_1(s)ds\\
                &=&\frac{n_1}{n^2h_n^2}\int_a^bK^2\left(\frac{t-s}{h_n}\right)(m^2(s)+\sigma_{\epsilon_1}^2)f_1(s)ds\\
                &=&\frac{1}{nh_n}\frac{n_1}{n}\int_{-\infty}^{\infty} I_{\left(-\frac{b-t}{h_n},-\frac{t-a}{h_n}\right)}(u)K^2(u)(m^2(t-h_nu) + \sigma_{\epsilon_1}^2)f_1(t-h_nu)du.
            \end{eqnarray*}
            By Assumptions~A1--A3 and that $m$ is bounded and continuous, the integrand is bounded and it tends to
            $\left(m^2(t) + \sigma_{\epsilon_1}^2\right)f_1(t)K^2(u)$
            as $n\rightarrow\infty$. It follows from the Dominated Convergence Theorem that the integral tends to $\left(m^2(t) + \sigma_{\epsilon_1}^2\right)f_1(t)\|K\|_2^2$,
            where
            $$\|K\|_2^2=\int_{-\infty}^{\infty}K^2(u)du,$$
            and consequently (by Assumption~A3) $E(T_{1n}) \rightarrow 0$ as $n\rightarrow\infty$.
            %dominated convergence theorem, and the continuity and boundedness of $m$, it follows that $$E(T_{1n})=O\left(\frac{1}{nh_n}\right).$$
            Next,
            \begin{eqnarray*}
                &&\hskip-25pt E(T_{2n})\\
                &=&\frac{n_2}{n^2h_n^2}\int_c^dK^2\left(\frac{t-g(s)}{h_n}\right)(m^2(g_0(s))+\sigma_{\epsilon_2}^2)f_2(s)dt\\
                &=&\frac{1}{nh_n}\frac{n_2}{n}\int_{\infty}^{\infty} I_{\left(-\left(\frac{g(d)-t}{h_n}\right),\left(\frac{t-g(c)}{h_n}\right)\right)}(u) K^2(u)(m^2(g_0\circ g^{-1}(t-h_nu))+\sigma_{\epsilon_2}^2)\\
                &&\times f_2\circ g^{-1}(t-h_nu)(g^{-1})'(t-h_nu)du.
%                    &=&O\left(\frac{1}{nh_n}\right).\notag
            \end{eqnarray*}
            It follows from a similar argument as that used for $T_{1n}$ that $E(T_{2n}) \rightarrow 0$ as $n\rightarrow\infty$. The third term simplifies as
            $$E(T_{3n}) = \frac{n_1(n_1-1)}{n^2h_n^2}\left\{\int_a^b K\left(\frac{t-u}{h_n}\right)m(u)f_1(u)du\right\}^2 = \left(1-\frac{1}{n_1}\right)E_{1n}^2. $$
            Likewise,
            \begin{eqnarray*}
                E(T_{4n}) &=& \left(1-\frac{1}{n_2}\right)E_{2n}^2,\\
                E(T_{5n}) &=& 2E_{1n} E_{2n}.
            \end{eqnarray*}
            and
            $$E(T_{3n})+E(T_{4n})+E(T_{5n}) = E^2(N_{n,t}(g))%-O\left(\frac{1}{n}\right).\notag
                - \frac{1}{n_1}E_{1n}^2 - \frac{1}{n_2}E_{2n}^2.$$
            It follows that
            $$V\left(N_{n,t}(g)\right)= %O\left(\frac{1}{nh_n}\right).\label{eq:VNn_conv}
                E\left(T_{1n}\right) + E\left(T_{2n}\right) - \left(\frac{1}{n_1}E_{1n}^2 + \frac{1}{n_2}E_{2n}^2\right).$$
            We have already shown that $E_{1n}$ and $E_{2n}$ converge to constants. Therefore,
            \begin{equation}
                V\left(N_{n,t}(g)\right) \rightarrow 0\quad\mbox{as }n\rightarrow\infty. \label{eq:VNn_conv}
            \end{equation}
            From~(\ref{eq:ENn_conv}) and~(\ref{eq:VNn_conv}) it follows that $N_{n,t}(g)\conv{P}N_t(g)$.

            The calculations for computing expectation and variance of $D_{n,t}(g)$ and their limiting values are similar to those of $N_{n,t}(g)$ except that the expressions do not involve components arising out of $y_{1i}$'s and $y_{2j}$'s (i.e.\ $m$, $\sigma_{\epsilon_1}$, and $\sigma_{\epsilon_2}$). These calculations show that, as $n$ tends to infinity
            \begin{eqnarray*}
                E(D_{n,t}(g)) &\rightarrow& D_t(g),\\
                V(D_{n,t}(g)) &\rightarrow& 0,
            \end{eqnarray*}
            and consequently $D_{n,t}(g)\conv{P} D_t(g)$, thus proving~\eqref{eq:NnDn_conv}. Further, by Assumptions~A1 and A3, $D_t(g)\ge\lambda_1f_1(t)>0$ for all $g$ in $\G_0$. The stated result follows from the Continuous Mapping Theorem.$\hfill\Box$
        \end{proof}
        %---------- End of Proof of Theorem 3.1: Point-wise convergence ----------%

        \medskip

        %%%%%%%%%%%%%% Proof of Theorem 3.2: Uniform convergence %%%%%%%%%%%%%%%%%%%
        \begin{proof}[Proof of Theorem~\ref{thm:uniform_conv}.]
            Note that for $g$, $\tilde{g}$ in $\G_0$ and any $t\in(a,b)$, $$|M_{n,t}(g)-M_t(g)|\le|M_{n,t}(g)-M_{n,t}(\tilde{g})| +|M_t(\tilde{g})-M_t(g)|+|M_{n,t}(\tilde{g})-M_t(\tilde{g})|.$$ Suppose $\epsilon>0$. Lemma~{\ref{lem:unif_cont_of_M}} proved below ensures that there exists $\delta_{\epsilon}>0$ such that
            \begin{eqnarray}
                \Delta(g,\tilde{g})<\delta_{\epsilon} \Rightarrow |M_t(g)-M_t(\tilde{g})|<\frac{\epsilon}{3}.
            \end{eqnarray}
            From Lemma~\ref{lem:oscillation_of_M} proved below, we have
            $$
            |M_{n,t}(\tilde{g})-M_{n,t}(g)|\le B_{n,t}(\tilde{g})\Delta(g,\tilde{g}),
            $$
            where $B_{n,t}(\tilde{g})$ is a quantity that converges in probability to~$B_t(\tilde{g})$, defined in that lemma. Therefore,
            \begin{eqnarray}\label{eq:thm:uniform_conv:1}
                P\left(B_{n,t}(\tilde{g})>\max\left\{\frac{\epsilon}{3\delta_\epsilon},2B_t(\tilde{g})\right\}\right)\rightarrow0.
            \end{eqnarray}
            Define $N_\delta(\tilde{g})=\{g:\Delta(g,\tilde{g})<\delta\}.$
            For given $\tilde{g}$, let
            \begin{eqnarray*}
                \delta(\tilde{g},\epsilon)=
                \left\{
                \begin{array}{lll}
                \min\left\{\frac{\epsilon}{6
                B_t(\tilde{g})},\delta_{\epsilon}\right\}&\mbox{if}& B_t(\tilde{g})>0\\
                \delta_{\epsilon}&\mbox{if}& B_t(\tilde{g})=0.
                \end{array}\right.
            \end{eqnarray*}
            Note that $\left\{N_{\delta(\tilde{g},\epsilon)}(\tilde{g}):\tilde{g}
            \in\G\right
            \}$ is an open cover of $\G_0$.
            As $\G_0$ is compact, there exists a finite
            sub-cover say
            $\left\{N_{\delta(\tilde{g}_j,\epsilon)}(\tilde{g}_j)\right\}_{j=1
            \ldots k_\epsilon}$, with
            $\G_0\subset\cup_{j=1}^{k_{\epsilon}}N_{\delta(\tilde{g}_j,\epsilon)}
            (\tilde {g}_j)$ for some finite $k_\epsilon$. It follows that
            \begin{eqnarray*}
                &&\hskip-20pt \sup_{g\in\G}|M_{n,t}(g)-M_t(g)|\\
                &\leq&\max_{j=1,\ldots,k_\epsilon}\sup_{
                g\in N_{\delta(\tilde{g}_j,\epsilon)}(\tilde{g}_j)}
                |M_{n,t}(g)-M_t(g)|\\
                &\leq&\max_{j=1,\ldots,k_\epsilon}\left\{
                \sup_{g\in N_{\delta(\tilde{g}_j,\epsilon)}(\tilde{g}_j)}
                |M_{n,t}(g)-M_{n,t}(\tilde{g}_j)|+\sup_{
                g\in N_{\delta(\tilde{g}_j,\epsilon)}(\tilde{g}_j)}
                |M_t(\tilde{g} _j)-M_t(g)|\right.\\
                &&+\left.\sup_{
                g\in N_{\delta(\tilde{g}_j,\epsilon)}(\tilde{g}_j)}
                |M_{n,t}(\tilde{g}_j)-M_t(\tilde {g}_j)|\right\}\\
                &\leq&\max_{j=1,\ldots,k_\epsilon}\left\{\delta(\tilde{g}_j,
                \epsilon)B_{n,t}(\tilde{g}_j)+\frac{\epsilon}{3}+|M_{n,t}(\tilde{g}_j)-M_t(\tilde{g}
                _j)|\right\}\\
                &\leq&\max_{j=1,\ldots,k_\epsilon}\delta(\tilde{g}_j,
                \epsilon)B_{n,t}(\tilde{g}_j)+\frac{\epsilon}{3}+\sum_{j=1}^{k_\epsilon}|M_{n,t}(\tilde{
                g}_j)-M_t(\tilde{g}_j)|.
            \end{eqnarray*}
            Thus,
            \begin{eqnarray}\label{eq:thm:uniform_conv:final}
                &&\hskip-25pt P\left\{\sup_{g\in\G}
                |M_{n,t}(g)-M_t(g)|>\epsilon\right\}\notag\\
                &\leq&
                P\left\{\max_{j=1,\ldots,k_{\epsilon}}\delta(\tilde { g }
                _j,\epsilon)B_{n,t}(\tilde{g}_j)>\frac{\epsilon}{3}\right\}+P\left\{\sum_{j=1}^{
                k_\epsilon}|M_{n,t}(\tilde{g}_j)-M_t(\tilde{g}_j)|>\frac{\epsilon}{3}\right\}\notag\\
                &\leq&\sum_{j=1}^{k_{\epsilon}}
                P\left\{B_{n,t}(\tilde{g}_j)>\frac{\epsilon}{3\delta(\tilde{g}_j,\epsilon)}\right\}+P\left\{\sum_{j=1}^{
                k_\epsilon}|M_{n,t}(\tilde{g}_j)-M_t(\tilde{g}_j)|>\frac{\epsilon}{3}\right\}.\hskip40pt\mbox{}
            \end{eqnarray}

            \noindent
            Each summand of the first term on the right hand side of
            (\ref{eq:thm:uniform_conv:final}) goes to zero by (\ref{eq:thm:uniform_conv:1}), while the second term goes
            to zero by Theorem~\ref{thm:pointwise_conv}. This completes the proof.
        \end{proof}
        %---------- End of Proof of Theorem 3.2: Uniform convergence ----------%

        \medskip

    %%%%%%%%%%%%%%%%% Lemma 4 %%%%%%%%%%%%%%%%%%%%
    \begin{lemma}\label{lem:unif_cont_of_M}
        Let $m$ and $\G_0$ be as described in Theorem~\ref{thm:uniform_conv}. Then under Assumptions~A1, A2, A3 and~A5, the functional $M_t$ in~(\ref{eq:M}) is uniformly continuous on $\G_0$ for any $t\in(a,b)$.
    \end{lemma}

    \noindent
    \begin{proof}[Proof.]
        Since $\G_0$ is compact, it is enough to establish pointwise continuity of $M_x$ on $\G_0$. Let $g$, $g_1$, $g_2$, $\ldots$ be functions in $\G_0$ such that $\lim_{k\rightarrow \infty}\Delta(g_k,g)=0$. The expression of $N_t(g)$ in~\eqref{eq:N} implies that
        $$N_t(g_k)=\xi m(t)f_1(t)+ (1-\xi)m\circ g_0\circ g_k^{-1}(t)f_2\circ g_k^{-1}(t)(g_k^{-1})'(t).$$
        It follows from Assumptions~A1, A2, A3 and Lemma~\ref{lem:gninv_pw_conv} (stated and proved below) that, as $k\rightarrow\infty$, the terms $m\circ g_0\circ g_k^{-1}(t)$, $f_2\circ g_k^{-1}(t)$ and $(g_k^{-1})'(t)$ on the right hand side converge, to $m\circ g_0\circ g^{-1}(t)$, $f_2\circ g^{-1}(t)$ and $(g^{-1})'(t)$, respectively, uniformly over all~$t\in [\alpha,\beta]$, where $\alpha$ and $\beta$ are as defined in that lemma. For $t\not\in[\alpha,\beta]$, both~$N_t(g_k)$ and~$N_t(g)$ are~0. Therefore,
        $N_t(g_k)$ converges to $N_t(g)$, i.e., the functional $N_t$ is continuous on $\G_0$ for all real~$t$. A similar argument shows that the functional $D_t$ defined in~\eqref{eq:D} is continuous on $\G_0$ as well.

        As $t$ is in the support of $f_1$, $f_1(t)>0$ (by Assumption~A1) and as $\xi>0$ (by Assumption~A5), $D_t(g)>0$ for all $g$ in $G$. This establishes that the functional $M_t$ defined in~(\ref{eq:M2}) is continuous on $\G_0$.
    \end{proof}

    \medskip

    %%%%%%%%%%%%%%%%% Lemmas 2-3 %%%%%%%%%%%%%%%%%%%%
    \begin{lemma}\label{lem:gninv_pw_conv}
        Let $\G_0$ be as in Theorem~\ref{thm:uniform_conv}. Further, let $g$, $g_1$, $g_2$, $\ldots$ be functions in $\G_0$ such that $g_k\rightarrow g$ uniformly. Let $\alpha=\inf\{g(c),g_1(c),g_2(c),\ldots\}$ and $\beta=\sup\{g(d),g_1(d),g_2(d),\ldots\}$. Then
        \begin{enumerate}
        \item[(a)] $g_k^{-1}\rightarrow g^{-1}$ uniformly on $[\alpha,\beta]$.
        \item[(b)] $(g_k^{-1})'\rightarrow (g^{-1})'$ uniformly on $[\alpha,\beta]$.
        \end{enumerate}
    \end{lemma}

    \noindent
    \begin{proof}[Proof.]
        Choose $\epsilon>0$. Since $g^{-1}$ has bounded second derivative, it has a bounded first derivative over the interval $[\alpha,\beta]$. If $B>0$ is an upper bound on the derivative, it follows from the Mean Value Theorem that $|g^{-1}(t_1)-g^{-1}(t_2)|<B|t_1-t_2|$ for all $t_1,t_2\in [\alpha,\beta]$. The uniform convergence of $g_1$, $g_2$, $\ldots$ to $g$ implies that there exists an integer $N_\epsilon$ such that $k\ge N_\epsilon$ implies $|g(u)-g_k(u)|<\epsilon/B$ for all $u\in[c,d]$. For any $t\in [\alpha,\beta]$, we can choose $u=g_k^{-1}(t)$, so that $$k>N_\epsilon\Rightarrow |g\circ g_k^{-1}(t)-t|=|g\circ g_k^{-1}(t)-g_k\circ g_k^{-1}(t)|<\frac{\epsilon}{B}.$$
        Thus, $k>N_\epsilon$ implies
        $$|g_k^{-1}(t)-g^{-1}(t)|=|g^{-1}\circ g\circ g_k^{-1}(t)-g^{-1}(t)| \le B|g\circ g_k^{-1}(t)-t|<\epsilon$$
        for all $t\in [\alpha,\beta]$, i.e., $g_k^{-1}\rightarrow g^{-1}$ uniformly on $[\alpha,\beta]$. This completes the proof of part~(a).

        To prove part (b), observe for any $t,t+h\in [\alpha,\beta]$, by Taylor series expansion,
        \begin{eqnarray*}
            g_k^{-1}(t+h)&=&g_k^{-1}(t)+h(g_k^{-1})'(t)+\frac{h^2}{2!}(g_k^{-1})''(t+\xi_1(t,h))\\
            g^{-1}(t+h)&=&g^{-1}(t)+h(g^{-1})'(t)+\frac{h^2}{2!}(g^{-1})''(t+\xi_2(t,h))
        \end{eqnarray*}
        where $\xi_1(t,h)$ and $\xi_2(t,h)$ are some numbers lying in between $t$ and $t+h$. Therefore,
        \begin{eqnarray}\label{eq:lem:deriv_conv:1}
            &&\hskip-30pt |(g_k^{-1})'(t)-(g^{-1})'(t)|\notag\\
            &\le&\frac{1}{h}|g_k^{-1}(t+h)-g^{-1}(t+h)| +\frac{1}{h}|g_k^{-1}(t)-g^{-1}(t)|\notag\\
            &&+\frac{h}{2!}\left\{|(g_k^{-1})''(t+\xi_1(t,h))|+|(g^{-1})''(t+\xi_2(t,h)|\right\}.
        \end{eqnarray}
        Let $C>0$ be a common upper bound for the second derivatives of inverses of all the functions in $\G_0$, over $[\alpha,\beta]$. Given $\epsilon>0$, choose $h<\frac{\epsilon}{3C}$. Then there exists $N_{\epsilon,h}$ such that $|g_k^{-1}(t)-g^{-1}(t)|<\frac{\epsilon h}{3}$ for $k>N_{\epsilon,h}$ and for all $t$. Therefore, from~(\ref{eq:lem:deriv_conv:1}) we have,
        $$
            k>N_{\epsilon,h}\Rightarrow |(g_k^{-1})'(t)-(g^{-1})'(t)|\le\frac{\epsilon}{3}+\frac{\epsilon}{3}+\frac{h}{2}2C <\epsilon\quad\forall\ t,
        $$
        which completes the proof of part (b).
    \end{proof}

    \medskip

    \medskip

     %%%%%%%%%%%%%%%%% Lemma 5 %%%%%%%%%%%%%%%%%%%%
    \begin{lemma}\label{lem:oscillation_of_M}
        Let $\G$ be as described in Theorem~\ref{thm:pointwise_conv} and $M_{n,t}$ be as in~(\ref{eq:Mn2}). Then, under Assumptions~A1, A2, A3, A4* and A5, for $t\in(a,b)$,
        \begin{enumerate}
        \item[(a)] $|M_{n,t}(\tilde{g})-M_{n,t}(g)|\le B_{n,t}(\tilde{g})\Delta(g,\tilde{g})$, where $g$, $\tilde{g}\in\G$ and  $$B_{n,t}(\tilde{g})=\frac{n_2}{n}\frac{C_{K'}}{c^2} \left\{D_{n,t}(\tilde{g})\frac{1}{n_2}\displaystyle \sum_{j=1}^{n_2}|y_{2j}|+N_{n,t}(\tilde{g})\right\},$$
            $N_{n,t}$ and $D_{n,t}$ being as defined in \eqref{eq:Nn} and \eqref{eq:Dn}, respectively.
        \item[(b)] $B_{n,t}(\tilde{g})\conv{P}B_t(\tilde{g})$, where $$B_t(\tilde{g})=(1-\xi)\frac{C_{K'}}{c^2}\left\{ D_t(\tilde{g})\mu_{|y_{2}|}+N_t(\tilde{g})\right\},$$
        \end{enumerate}
        $N_t$ and $D_t$ being as defined in \eqref{eq:N} and \eqref{eq:D}, respectively, and $\mu_{|y_2|}=E(|y_{2j}|)$ for all $j=1, 2, \ldots\, n_2$.
    \end{lemma}

    \noindent
    \begin{proof}[Proof.]
        Let the positive real numbers $c$, $C_{K'}$ be such that $0<c\le K(t)$ and $|K'(t)|\le C_{K'}$ (by Assumption~A4*). Therefore, $D_{n,t}(g)\ge \frac{c}{h}$. Consequently, for any given $g, \tilde{g}\in\G$, we have
        \begin{eqnarray*}
            &&\hskip-30pt |M_{n,t}(g)-M_{n,t}(\tilde{g})|\\
            &\le& \frac{h^2}{c^2}\bigg\{D_{n,t}(\tilde{g})|N_{n,t}(g)-N_{n,t}(\tilde{g})| +N_{n,t}(\tilde{g})|D_{n,t}(\tilde{g})-D_{n,t}(g)|\bigg\}.
        \end{eqnarray*}
        For an upper bound on the first term in the square bracket, note that
        $$|N_{n,t}(g)-N_{n,t}(\tilde{g})|= \frac{1}{nh}\left|\displaystyle\sum_{j=1}^{n_2}y_{2j}\left\{K\left(\frac{t-g(s_j)}{h}\right )-K\left (\frac{t-\tilde{g}(s_j)}{h}\right ) \right\}\right|.$$
        By using the Mean Value Theorem and the bound of $K'(\cdot)$, we have
        $$|N_{n,t}(g)-N_{n,t}(\tilde{g})|\le \frac{1}{n}\times\frac{C_{K'}}{h^2}\Delta(g,\tilde{g})\displaystyle\sum_{j=1}^{n_2}|y_{2j}|$$.
        Similarly,
        $$|D_{n,t}(\tilde{g})-D_{n,t}(g)|\le \frac{n_2}{n}\times\frac{C_{K'}}{h^2}\Delta(g,\tilde{g})$$.
        Therefore,
        $$|M_{n,t}(g)-M_{n,t}(\tilde{g})|\le \frac{n_2}{n}\times\frac{C_{K'}}{c^2}\left\{D_{t,x}(\tilde{g})\frac{1}{n_2}\displaystyle \sum_{j=1}^{n_2}|y_{2j}|+N_{n,t}(\tilde{g})\right\}\Delta(g,\tilde{g}).$$
        This completes the proof of part (a).

        As $s_j$'s and $\epsilon_{2j}$'s are i.i.d., the $|y_{2j}|$'s are also i.i.d.\ and by Assumptions~A1 and A2, their mean, $\mu_{|y_2|}$, is finite. Therefore, by the strong law of large numbers we have, as $n_2\rightarrow \infty$
        $\frac{1}{n_2}\sum_{j=1}^{n_2}|y_{2j}|\rightarrow \mu_{|y_2|}$ almost surely. The proof of part~(b) is then completed from Assumptions~A2, A3 and~A5 and~\eqref{eq:NnDn_conv}.
    \end{proof}

    \medskip

    %%%%%%%%%%%%%%%%% Proof of Theorem 3.3: Consistency %%%%%%%%%%%%%%%%%%%%
    \begin{proof}[Proof.]
         Suppose $\epsilon>0$. Lemma~{\ref{lem:unif_cont_of_M}} ensures that there exist $\delta_{\epsilon}>0$ such that,
        \begin{equation}
        \Delta(g,\tilde{g})<\delta_{\epsilon} \Rightarrow \Delta(M_t(g),M_t(\tilde{g}))<\epsilon\label{eq:thm:consistency_1}.
        \end{equation}
        Now consider
        \begin{eqnarray}
            &&\hskip-25pt P\{|M_{n,t}(\hat{g}_n)-M_t(g_0)|>2\epsilon\}\notag\\
            &\le & P\left\{|M_{n,t}(\hat{g}_n)-M_t(\hat{g}_n)|>\epsilon\right\}+
            P\left\{|M_t(\hat{g}_n)-M_t(g_0)|>\epsilon\right\}\notag\\
            &\le & P\left\{|M_{n,t}(\hat{g}_n)-M_t(\hat{g}_n)|>\epsilon\right\}\notag\\
            &&+P\left\{|M_t(\hat{g}_n)-M_t(g_0)|>\epsilon,\Delta(\hat{g}_n,g_0)<\delta_{\epsilon}\right\} +P\left\{\Delta(\hat{g}_n,g_0)\ge\delta_{\epsilon}\right\}.\hskip25pt\mbox{}\label{eq:thm:consistency_prob_ineq}
        \end{eqnarray}
        By observing that $$|M_{n,t}(\hat{g}_n)-M_t(\hat{g}_n)|\le\sup_{g\in\G}|M_{n,t}(g)-M_t(g)|,$$ and using Theorem~\ref{thm:uniform_conv}, the first term in~(\ref{eq:thm:consistency_prob_ineq}) can be seen to go to zero as $n\rightarrow \infty$. The second term vanishes by~(\ref{eq:thm:consistency_1}). As $\hat{g}_n$ is a consistent estimator of $g_0$, the third term also goes to zero. This completes the proof.
    \end{proof}
    %---------- Proof of Theorem 3.3: Consistency ----------%

    \medskip

        %%%%%%%%%%%%%%%%% Proof of Theorem 3.4: MSE %%%%%%%%%%%%%%%%%%%%

        \begin{proof}[Proof of Theorem~\ref{thm:mse_g0}.]
        Instead of directly considering the second moment of the estimation error $m_{n,g_0}(t)-m(t)$, we shall show that the expression on the right hand side of \eqref{eq:thm:mse_g0:main} is the second moment of
        \begin{equation}
            H_n(t) = \frac{r_{n,g_0}(t)-m(t)f_{n,g_0}(t)}{f_{g_0}(t)},\label{eq:Hn}
        \end{equation}
        where $r_{n,g_0}(t)$ and $f_{n,g_0}(t)$ are the numerator and denominator of $m_{n,g}(t)$ in~\eqref{eq:nw_type_functional}, respectively, for $g=g_0$, and then show that $H_n(t)$ is not very far from the estimation error.             Clearly,
        \begin{equation}
            E(H_n^2(t))=\frac{E(r_{n,g_0}^2(t))+ m^2(t)E(f_{n,g_0}^2(t))-2m(t)E(r_{n,g_0}(t)f_{n,g_0}(t))}{f_{g_0}^2(t)}.\label{eq:thm:linearization_L2:L2}
        \end{equation}
        From the expression of $r_{n,g_0}(t)$, we have
        \begin{eqnarray*}
            r_{n,g_0}^2(t)&=&\frac{1}{n^2h_n^2}\left\{\sum_{i=1}^{n_1} K^2\left(\frac{t-t_i}{h_n}\right) y_{1i}^2 + \sum_{j=1}^{n_2} K^2\left(\frac{t-g_0(s_j)}{h_n}\right) y_{2j}^2\right.\\
                &&+ \sum_{i=1}^{n_1}\sum_{\substack {i'=1\\ (\neq i)}}^{n_1} K\left(\frac{t-t_i}{h_n}\right)K\left(\frac{t-t_{i'}}{h_n}\right)y_{1i}y_{1i'}\\
                &&+ \sum_{j=1}^{n_2}\sum_{\substack {j'=1\\ (\neq j)}}^{n_2} K\left(\frac{t-g_0(s_j)}{h_n}\right) K\left(\frac{t-g_0(s_{j'})}{h_n}\right)y_{2j}y_{2j'}\\
                &&\left.+ \sum_{i=1}^{n_1}\sum_{j=1}^{n_2} K\left(\frac{t-t_i}{h_n}\right)K\left(\frac{t-g_0(s_{j})}{h_n}\right)y_{1i}y_{2j}\right\}.
        \end{eqnarray*}
        Consequently,
        $$E(r_{n,g_0}^2(t))= \frac{1}{nh_n}\frac{n_1}{n}I_1 + \frac{1}{nh_n}\frac{n_2}{n}I_2 + \frac{n_1(n_1{-}1)}{n^2}I_3^2 + \frac{n_2(n_2{-}1)}{n^2}I_4^2 + \frac{2n_1n_2}{n^2}I_3I_4,$$
        where
        \begin{eqnarray*}
            I_1&=&\int s_1^2(t-h_nu)f_1(t-h_nu)K^2(u)du,\\
            I_2&=&\int s_2^2(g_0^{-1}(t-h_nu))f_2\circ g_0^{-1}(t-h_nu)(g_0^{-1})'(t-h_nu)K^2(u)du,\\
            I_3&=&\int m(t-h_nu)f_1(t-h_nu)K(u)du,\\
            I_4&=&\int m(t-h_nu)f_2\comp g_0^{-1}(t-h_nu)(g_0^{-1})'(t-h_nu)K(u)du,\\
            s_1^2(u)&=&E(y_{1i}^2|t_i=u) \quad \mbox{and }s_2^2(u)=E(y_{2j}^2|s_j=u).
        \end{eqnarray*}
        All these integrals are functions of $t$. Similarly, it can be shown that
        \begin{eqnarray*}
            E(f_{n,g_0}^2(t))&=& \frac{1}{nh_n}\frac{n_1}{n}I_5 + \frac{1}{nh_n}\frac{n_2}{n}I_6 + \frac{n_1(n_1-1)}{n^2}I_7^2\\
                &&+ \frac{n_2(n_2-1)}{n^2}I_8^2 + \frac{2n_1n_2}{n^2}I_7I_8,\\
            E(r_{n,g_0}(t)f_{n,g_0}(t))&=& \frac{1}{nh_n}\frac{n_1}{n}I_9 + \frac{1}{nh_n}\frac{n_2}{n}I_{10} + \frac{n_1(n_1-1)}{n^2}I_3I_7 \\
                &&+\frac{n_2(n_2-1)}{n^2}I_4I_8 + \frac{n_1n_2}{n^2}I_3I_8 + \frac{n_1n_2}{n^2}I_4I_7 ,
        \end{eqnarray*}
        where
        \begin{eqnarray}
            I_5&=&\int f_1(t-h_nu)K^2(u)du,\notag\\
            I_6&=&\int f_2\comp g_0^{-1}(t-h_nu)(g_0^{-1})'(t-h_nu)K^2(u)du,\notag\\
            I_7&=&\int f_1(t-h_nu)K(u)du,\notag\\
            I_8&=&\int f_2\comp g_0^{-1}(t-h_nu)(g_0^{-1})'(t-h_nu)K(u)du,\notag\\
            I_9&=&\int m(t-h_nu)f_1(t-h_nu)K^2(u)du,\notag\\
            I_{10}&=&\int m(t-h_nu)f_2\comp g_0^{-1}(t-h_nu)(g_0^{-1})'(t-h_nu)K^2(u)du,\notag
        \end{eqnarray}
        which are also functions of $t$. Therefore, the numerator in~(\ref{eq:thm:linearization_L2:L2}),
        \begin{eqnarray}\label{eq:thm:linearization_L2:L2_num}
            &&\hskip-30pt E(r_{n,g_0}^2(t)) + m^2(t) E(f_{n,g_0}^2(t)) - 2m(t) E(r_{n,g_0}(t)f_{n,g_0}(t))\notag\\
            &=&\frac{1}{nh_n}\left\{\frac{n_1}{n}(I_1+m^2(t)I_5-2m(t)I_9) + \frac{n_2}{n}(I_2+m^2(t)I_6-2m(t)I_{10})\right\}\notag\\
            &&\quad +\left\{\frac{n_1}{n}(I_3-m(t)I_7)+\frac{n_2}{n}(I_4-m(t)I_8)\right\}^2\notag\\
            &&\quad-\frac{1}{n}\left\{\frac{n_1}{n}(I_3-m(t)I_7)^2+\frac{n_2}{n}(I_4-m(t)I_8)^2\right\}.
        \end{eqnarray}

        We first show that, $I_1$ tends to $s_1^2(t)f_1(t)\|K\|_2^2$ as $n\rightarrow\infty$. Fix $\epsilon>0$. As $s_1^2$ and $f_1$ are continuous functions, we may choose $\delta>0$ so that $$\sup_{|u|\le\delta}\left|s_1^2(t-u)f_1(t-u)-s_1^2(t)f_1(t)\right|\|K\|_2^2<\frac\epsilon3.$$  It follows that
        \begin{eqnarray*}
            &&\hskip-35pt \Big|I_1 - s_1^2(t)f_1(t)\|K\|_2^2\Big|\\
                &\!\!=&\left|\int_{-\infty}^{\infty} K^2(u)\left\{s_1^2(t-h_nu)f_1(t-h_nu)-s_1^2(t)f_1(t)\right\}du\right|\\
                &\!\!\le&\int_{|u|\le\delta}\frac{1}{h_n}K^2\left(\frac{u}{h_n}\right)\left|s_1^2(t-u)f_1(t-u)-s_1^2(t)f_1(t)\right|du\\
                    &&+\int_{|u|>\delta}\frac{1}{h_n}|u|K^2\left(\frac{u}{h_n}\right)\frac{s_1^2(t-u)f_1(t-u)}{|u|}du\\
                    &&+s_1^2(t)f_1(t)\int_{|u|>\delta}\frac{1}{h_n}K^2\left(\frac{u}{h_n}\right)du\\
                &\!\!\le&\sup_{|u|\le\delta}\left|s_1^2(t-u)f_1(t-u)-s_1^2(t)f_1(t)\right|\|K\|_2^2\\
                    &&\quad+\frac{M_KM_1}{\delta}\sup_{|u|>\frac{\delta}{h_n}}|uK(u)|+s_1^2(t)f_1(t)\int_{|u|>\frac{\delta}{h_n}}K^2(u)du,
        \end{eqnarray*}
        where $M_1$ (by Assumptions~A1$'$ and~A2$'$) and $M_K$ (by Assumption~A4$'$) are upper bounds of $E(y_{1i}^2)$ and $K$ respectively. We may select $n$ sufficiently large so that each of the last two terms in the above expression is smaller than $\frac\epsilon3$ (by Assumption~A4$'$). This proves our assertion. Thus, we can write
        $$\frac{n_1}{n}I_1 = \frac{n_1}{n}s_1^2(t)f_1(t)\|K\|_2^2+o(1)=\frac{n_1}{n}\left(\sigma_{\epsilon_1}^2+m^2(t)\right)f_1(t)\|K\|_2^2+o(1).$$
        By a similar argument,
        $$\frac{n_2}{n}I_2= \frac{n_2}{n} \left(\sigma_{\epsilon_2}^2+m^2(t)(g_0^{-1}(t))\right) f_2\circ g_0^{-1}(t)({g_0^{-1}})'(t)\|K\|_2^2+o(1).$$

        We now evaluate the integrals $I_3$ to $I_{10}$. Taylor expansions of $m(t-h_nu)$ and $f_1(t-h_nu)$ around the values of these function at $t$, up to second order derivatives, give
        \begin{eqnarray*}
            I_3&=&\int_{-\infty}^{\infty}m(t-h_nu)f_1(t-h_nu)K(u)du\\
                &=&\int_{-\infty}^{\infty}\left\{m(t)-h_num'(t)+\frac{(h_nu)^2}{2}m''(t_1)\right\}\\
                    &&\ \ \times\left\{f_1(t)-h_nuf_1'(t)+\frac{(h_nu)^2}{2}f_1''(t_2)\right\}K(u)du
        \end{eqnarray*}
        where $t_1=t_{1}(t,h_n,u,m)$ and $t_2=t_{2}(t,h_n,u,f_1)$ lie between $t$ and $t-h_nu$ and tend to $t$ as $n$ tends to infinity. Thus
        \begin{eqnarray}
            I_3&=&m(t)f_1(t)+\frac{h_n^2}{2}\left\{f_1(t)\int_{-\infty}^{\infty} m''(t_1)u^2K(u)du+2m'(t)f_1'(t)\mu_2(K)\right.\notag\\
                &&\quad\quad\left.+m(t)\int_{-\infty}^{\infty}f_1''(t_2)u^2K(u)du\right\}\notag\\
                &&\quad-\frac{h_n^3}{2}\left\{f_1'(t)\int_{-\infty}^{\infty}m''(t_1)u^3K(u)du +m'(t)\int_{-\infty}^{\infty}f_1''(t_2)u^3K(u)du\right\}\notag\\
                &&\quad+\frac{h_n^4}{4}\int_{-\infty}^{\infty}m''(t_1)f_1''(t_2)u^4K(u)du.\label{eq:thm:linearization_L2:e3n_integral}
        \end{eqnarray}
       Consider the first integral in~(\ref{eq:thm:linearization_L2:e3n_integral}). Since the kernel $K$ is assumed to be compactly supported, this integral is effectively over a finite interval. It follows from Assumption~A5$'$ and the continuity of the function $m''$ that the integrand converges to $m''(t)u^2K(u)$ . Assumptions~A2$'$~and~A4$'$~imply that the integrand is bounded by an integrable function on this support. It, therefore, follows from the Dominated Convergence Theorem that
        $$M_{2n}=\int_{-\infty}^{\infty} m''(t_1)u^2K(u)du - m''(t)\mu_2(K) = o(1).$$
        By applying a similar argument for the other integrals in~\eqref{eq:thm:linearization_L2:e3n_integral} and consolidating the terms, we obtain
        \begin{eqnarray*}
            I_3\!\!&=\!\!&m(t)f_1(t)+\frac{h_n^2}{2}\Big\{(m(t)f_1(t))''\mu_2(K)+f_1(t)M_{2n}+m(t)F_{2n}\Big\}\\
            &&\ \ -\frac{h_n^3}{2}\Big\{f_1'(t)M_{3n}+m'(t)F_{3n}\Big\}+\frac{h_n^4}{4}m''(t)f_1''(t)\int_\infty^\infty u^4K(u)du+o(h_n^4),
        \end{eqnarray*}
        where
        \begin{eqnarray}
            M_{3n} &=& \int_{-\infty}^{\infty} m''(t_1)u^3K(u)du = o(1),\notag\\
            F_{2n} &=& \int_{-\infty}^{\infty} f_1''(t_2)u^2K(u)du - f_1''(t)\mu_2(K) = o(1),\notag\\
            F_{3n} &=& \int_{-\infty}^{\infty} f_1''(t_2)u^3K(u)du = o(1).\notag
        \end{eqnarray}

        Similarly, from Assumptions~A1$'$-A5$'$\ and Taylor expansions of $m(t-h_nu)$, $f_1(t-h_nu)$, $f_2\circ g_0^{-1}(t-h_nu)$ and $(g_0^{-1})'(t-h_nu)$ around the values of these functions at $t$, up to second order derivatives, we obtain
        \begin{eqnarray*}
            I_4%&=&\int m(t-h_nu)f_2\comp g_0^{-1}(t-h_nu){g_0^{-1}}'(t-h_nu)K(u)du\\
                &=&m(t)f_2\circ g_0^{-1}(t)(g_0^{-1})'(t)+\frac{h_n^2}{2}\left\{(m(t)f_2\circ g_0^{-1}(t)(g_0^{-1}(t))')''\mu_2(K)\right.\\
                    &&\left.+f_2\circ g_0^{-1}(t)(g_0^{-1}(t))'M_{2n} +m(t)f_2\circ g_0^{-1}(t)G_{2n} +m(t)(g_0^{-1}(t))'F_{2n}^\ast \right\}\\
                    &&-\frac{h_n^3}{2}\left[\left\{f_2\circ g_0^{-1}(t) (g_0^{-1}(t))''+(f_2\circ g_0^{-1}(t))'(g_0^{-1}(t))'\right\}M_{3n}\right.\\
                        &&\quad +\left\{m(t)(f_2\circ g_0^{-1}(t))'+m'(t)f_2\circ g_0^{-1}(t)\right\}G_{3n}\\
                        &&\quad\left. +\left\{m(t)(g_0^{-1}(t))''+m'(t)(g_0^{-1}(t))'\right\}F_{3n}^\ast\right]\\
                    &&+\frac{h_n^4}{4}\left\{m''(t)(f_2\circ g_0^{-1}(t)(g_0^{-1}(t))')''+2m'(t)((f_2\circ g_0^{-1}(t))'(g_0^{-1}(t))'')'\right.\\
                        &&\quad\left.+m(t)(f_2\circ g_0^{-1}(t))''(g_0^{-1}(t))'''\right\}\int u^4K(u)du+o(h_n^4),\\
            I_5%=\frac{1}{nh_n}\int f_1(t-h_nu)K^2(u)du
                &=&f_1(t)\|K\|_2^2+O\left(h_n^2\right),\\
            I_6&=&f_2\comp g_0^{-1}(t)(g_0^{-1}(t))'\|K\|_2^2+O\left(h_n^2\right),\\
            I_7%=\int f_1(t-h_nu)K(u)du
                &=&f_1(t)+\frac{h_n^2}{2}\left\{f_1''(t)\mu_2(K)+F_{2n}\right\},\\
            I_8%&=&\int f_2\comp g_0^{-1}(t-h_nu)(g_0^{-1})'(t-h_nu)K(u)du\\
                &=&f_2\comp g_0^{-1}(t)(g_0^{-1}(t))'\\
                &&+\frac{h_n^2}{2}\left\{(f_2\comp g_0^{-1}(t)(g_0^{-1}(t))')''\mu_2(K)+ f_2\comp g_0^{-1}(t)G_{2n}+(g_0^{-1}(t))'F_{2n}^\ast\right\}\\
                    &&-\frac{h_n^3}{2}\left\{(f_2\comp g_0^{-1}(t))'G_{3n}+(g_0^{-1}(t))''F_{3n}^\ast\right\}\\
                    &&+\frac{h_n^4}{4}(f_2\comp g_0^{-1}(t))''(g_0^{-1}(t))'''\int u^4K(u)du+o(h_n^4),\\
            I_9%=\frac{1}{nh_n}\int m(t-h_nu)f_1(t-h_nu)K^2(u)du
                &=&m(t)f_1(t)\|K\|_2^2+O\left(h_n^2\right),\\
            I_{10}%=\frac{1}{nh_n}\int m(t-h_nu)f_1(t-h_nu)K^2(u)du
                &=&m(t)f_2\comp g_0^{-1}(t) (g_0^{-1}(t))'\|K\|_2^2+O\left(h_n^2\right),
        \end{eqnarray*}
        where
        \begin{eqnarray*}
            G_{2n} &=& \int_{-\infty}^{\infty} (g_0^{-1})'''(t_2^\ast)u^2K(u)du - {g_0^{-1}}'''(t)\mu_2(K) = o(1),\\
            G_{3n} &=& \int_{-\infty}^{\infty} (g_0^{-1})'''(t_2^\ast)u^3K(u)du = o(1),\\
            F_{2n}^\ast &=& \int_{-\infty}^{\infty} (f_2 \comp g_0^{-1})''(t_1^\ast)u^2K(u)du - (f_2 \comp g_0^{-1})''(t)\mu_2(K) = o(1),\\
            F_{3n}^\ast &=& \int_{-\infty}^{\infty} (f_2 \comp g_0^{-1})''(t_1^\ast)u^3K(u)du = o(1),
        \end{eqnarray*}
        and $t_1^\ast=t_1^\ast(t,h_n,u,f_2,g_0)$ and $t_2^\ast=t_2^\ast(t,h_n,u,g_0)$ lie between $t$ and $t-h_nu$ and tend to $t$ as $n$ tends to infinity. Therefore,
        \begin{eqnarray*}
            &&\hskip-30pt\frac{1}{nh_n}\times\frac{n_1}{n}(I_1+m^2(t)I_5-2m(t)I_9)\\ &=& \frac{1}{nh_n}\times\frac{n_1}{n}f_1(t)\sigma_{\epsilon_1}^2\|K\|_2^2 + o\left(\frac{1}{nh_n}\right) + O\left(\frac{h_n}{n}\right)\\
                &=&\frac{1}{nh_n}\times\frac{n_1}{n}f_1(t)\sigma_{\epsilon_1}^2\|K\|_2^2 + o\left(\frac{1}{nh_n}\right),\\
            &&\hskip-30pt\frac{1}{nh_n}\times\frac{n_2}{n}(I_2+m^2(t)I_6-2m(t)I_{10})\\
    		&=& \frac{1}{nh_n}\times\frac{n_2}{n}f_2\comp g_0^{-1}(t)(g_0^{-1}(t))' \sigma_{\epsilon_2}^2\|K\|_2^2 + o\left(\frac{1}{nh_n}\right).
        \end{eqnarray*}
        Consequently, by Assumption~A5$'$,
        \begin{eqnarray}\label{eq:thm:linearization_L2:L2_num:1}
            &&\hskip-30pt \frac{1}{nh_n}\left\{\frac{n_1}{n}(I_1+m^2(t)I_5-2m(t)I_9) + \frac{n_2}{n}(I_2+m^2(t)I_6-2m(t)I_{10})\right\}\notag\\
            &&=\frac{\xi f_1(t)\sigma_{\epsilon_1}^2+ (1-\xi) f_2\comp g_0^{-1}(t)(g_0^{-1}(t))'\sigma_{\epsilon_2}^2}{nh_n}\|K\|_2^2 + o\left(\frac{1}{nh_n}\right).
        \end{eqnarray}
        From the expressions of $I_3$, $I_4$, $I_7$, and $I_8$ above, we have
        \begin{eqnarray*}
            &&\hskip-25pt I_3-m(t)I_7\\
    		&=& \frac{h_n^2}{2}\left\{\left((m(t)f_1(t))'' - m(t)f_1''(t)\right)\mu_2(K) + f_1(t)M_{2n} \right\}\\
                        &&- \frac{h_n^3}{2}\left\{f_1'(t)M_{3n} + m'(t)F_{3n}\right\} + \frac{h_n^4}{4}m''(t)f_1''(t)\int u^4K(u)du+ o(h_n^4),\\
                    &&\hskip-25pt I_4-m(t)I_8\\
    		&=& \frac{h_n^2}{2}\Bigg\{\left((m(t)f_2\comp g_0^{-1}(t)(g_0^{-1}(t))')'' - m(t)\left(f_2\comp g_0^{-1}(t)(g_0^{-1}(t))'\right)''\right)\mu_2(K)\\
                &&\qquad + f_2\comp g_0^{-1}(t)(g_0^{-1}(t))'M_{2n} \Bigg\}\\
            && -\frac{h_n^3}{2}\bigg\{m'(t)f_2\comp g_0^{-1}(t)G_{3n} + m'(t)(g_0^{-1}(t))'F_{3n}^\ast\notag\\
                &&\qquad +\left(f_2\comp g_0^{-1}(t)(g_0^{-1}(t))'' + \left(f_2\comp g_0^{-1}(t)\right)'(g_0^{-1}(t))'\right)M_{3n}\bigg\}\\
            && + \frac{h_n^4}{4}\left\{m''(t)\left(f_2\comp g_0^{-1}(t)(g_0^{-1}(t))'\right)'' + 2m'(t)\left(\left(f_2\comp g_0^{-1}(t)\right)'(g_0^{-1}(t))''\right)'\right\}\\
            &&\qquad \times\int u^4K(u)du+ o(h_n^4).
        \end{eqnarray*}
        Therefore,
        \begin{eqnarray}\label{eq:thm:linearization_L2:L2_num:2}
            &&\hskip-25pt \left\{\frac{n_1}{n}(I_3-m(t)I_7)+\frac{n_2}{n}(I_4-m(t)I_8)\right\}^2\notag\\
            &&=\frac{h_n^4}{4}\Bigg\{\frac{n_1}{n}\left(\left(m(t)f_1(t)\right)'' - m(t)f_1''(t)\right)\notag\\
                &&\quad + \frac{n_2}{n}\left(\left(m(t)f_2\comp g_0^{-1}(t)(g_0^{-1}(t))'\right)'' - m(t)\left(f_2\comp g_0^{-1}(t)(g_0^{-1}(t))'\right)''\right)\Bigg\}^2\notag\\
		        &&\qquad\times\mu^2(K)+o(h_n^4)\notag\\
                &&=\frac{h_n^4}{4}\Bigg\{\frac{n_1}{n}\left(m''(t)f_1(t) + 2m'(t)f_1'(t)\right)\notag\\
                    &&\quad + \frac{n_2}{n}\left(m''(t)f_2\comp g_0^{-1}(t)(g_0^{-1}(t))' + 2 m'(t)\left(f_2\comp g_0^{-1}(t)(g_0^{-1}(t))'\right)'\right)\Bigg\}^2\notag\\
		&&\qquad\times\mu^2(K)+o(h_n^4)\notag\\
                &&=\frac{h_n^4}{4}\Bigg\{m''(t)\left(\frac{n_1}{n}f_1(t)+\frac{n_2}{n}f_2\comp g_0^{-1}(t)(g_0^{-1}(t))'\right)\notag\\
                    &&\quad + 2m'(t)\left(\frac{n_1}{n}f_1'(t)+\frac{n_2}{n}\left(f_2\comp g_0^{-1}(t)(g_0^{-1}(t))'\right)'\right)\Bigg\}^2\mu^2(K)+o(h_n^4)\notag\\
                &&=\frac{h_n^4}{4}\left\{m''(t)f_{g_0}(t) + 2 m'(t)f_{g_0}'(t) + o(1) \right\}^2\mu^2(K)+o(h_n^4)\notag\\
                &&=\frac{h_n^4}{4}f_{g_0}^2(t)\left(m''(t) + 2 \frac{m'(t)f_{g_0}'(t)}{f_{g_0}(t)}\right)^2\mu^2(K)+o(h_n^4).
        \end{eqnarray}
        Further,
        \begin{eqnarray*}
                    (I_3 - m(t)I_7)^2
    		&=& \frac{h_n^4}{4}\left\{\left(\left(m(t)f_1(t)\right)'' - m(t)f_1''(t)\right)\right\}^2\mu_2^2(K) + o(h_n^4),\\
                    (I_4 - m(t)I_8)^2
    		&=& \frac{h_n^4}{4}\left\{\left(m(t)f_2\comp g_0^{-1}(t)(g_0^{-1}(t))'\right)''\right.\\
    		&&\qquad\left. - m(t)\left(f_2\comp g_0^{-1}(t)(g_0^{-1}(t))'\right)''\right\}^2\mu_2^2(K) + o(h_n^4),
        \end{eqnarray*}
        and, consequently,
        \begin{equation}\label{eq:thm:linearization_L2:L2_num:3}
            \frac{1}{n}\left\{\frac{n_1}{n}(I_3-m(t)I_7)^2+\frac{n_2}{n}(I_4-m(t)I_8)^2\right\} = o(h_n^4).
        \end{equation}

        By combining~(\ref{eq:thm:linearization_L2:L2_num:1}--\ref{eq:thm:linearization_L2:L2_num:3}), (\ref{eq:thm:linearization_L2:L2_num}) and~(\ref{eq:thm:linearization_L2:L2}), we obtain
        \begin{eqnarray}\label{eq:Hn2moment}
           E(H_n^2(t))&\!\!=&\!\!\frac{\xi f_1(t)\sigma_{\epsilon_1}^2+ (1-\xi) f_2\circ g_0^{-1}(t)(g_0^{-1})'(t)\sigma_{\epsilon_2}^2}{nh_n}\times\frac{\|K\|_2^2}{f_{g_0}^2(t)}\notag\\
            &&\!\!+\frac{h_n^4}{4}\left[m''(t)+\frac{2m'(t)f_{g_0}'(t)}{f_{g_0}(t)}\right]^2\mu_2^2(K)+o\left(h_n^4+\frac{1}{nh_n}\right),\quad\mbox{}
        \end{eqnarray}
        as claimed at the beginning of this proof.

        We now return to the estimation error of $ m_{n,g_0}(t)$. Let us write it as
        \begin{eqnarray*}
            m_{n,g_0}(t)-m(t)&=&\frac{r_{n,g_0}(t)-m(t)f_{n,g_0}(t)}{f_{g_0}(t)}\\
&&\qquad+(m_{n,g_0}(t)-m(t))\left(\frac{f_{g_0}(t)-f_{n,g_0}(t)}{f_{g_0}(t)}\right).
        \end{eqnarray*}
        This decomposition implies
        \begin{eqnarray*}
           &&\hskip-30pt |m_{n,g_0}(t)-m(t)-H_n(t)|\\
&=&\left|m_{n,g_0}(t)-m(t)-\frac{r_{n,g_0}(t)-m(t)f_{n,g_0}(t)}{f_{g_0}(t)}\right|\\
	 &=&\frac{1}{f_{g_0}(t)}\left|(m_{n,g_0}(t)-m(t))(f_{n,g_0}(t)-f_{g_0}(t))\right|\\
            &\le&\frac{1}{2f_{g_0}(t)}\left\{(m_{n,g_0}(t)-m(t))^2+(f_{n,g_0}(t)-f_{g_0}(t))^2\right\},
        \end{eqnarray*}
        and consequently,
        \begin{eqnarray}
            &&\hskip-30pt E|m_{n,g_0}(t)-m(t)-H_n(t)|\notag\\
	&\le& \frac{1}{2f_{g_0}(t)}\left\{E(m_{n,g_0}(t)-m(t))^2+E(f_{n,g_0}(t)-f_{g_0}(t))^2\right\}.\label{eq:thm:linearization_order:1}
        \end{eqnarray}
        In order to evaluate $E(m_{n,g_0}-m)^2$  in~\eqref{eq:thm:linearization_order:1}, let us write
        $$m_{n,g_0}(t)=\frac{r_{n,g_0}(t)}{f_{g_0}(t)}+\frac{1}{f_{g_0}(t)}\times\frac{r_{n,g_0}(t)}{f_{n,g_0}(t)}(f_{g_0}(t)-f_{n,g_0}(t)).$$
        Consequently,
        \begin{eqnarray*}
           &&\hskip-30pt  (m_{n,g_0}(t)-m(t))^2\\
	 &\le&\frac{1}{f_{g_0}^2(t)}\bigg\{(r_{n,g_0}(t)-m(t)f_{g_0}(t))^2+\left(\frac{r_{n,g_0}(t)}{f_{n,g_0}(t)}\right)^2(f_{n,g_0}(t)-f_{g_0}(t))^2\\
            &&\quad +2\left|\frac{r_{n,g_0}(t)}{f_{n,g_0}(t)}\right||(r_{n,g_0}(t)-m(t)f_{g_0}(t))(f_{n,g_0}(t)-f_{g_0}(t))|\bigg\}.
        \end{eqnarray*}
        It follows from Assumptions~A1$'$ and~A2$'$\ that there is $C>0$ such that $|r_{n,g_0}(t)/f_{n,g_0}(t)|\le C$ almost surely. Therefore, we can simplify the above bound as
        \begin{eqnarray*}
           &&\hskip-30pt (m_{n,g_0}(t)-m(t))^2\\
     &\le&\frac{1}{f_{g_0}^2(t)}\big\{(r_{n,g_0}(t)-m(t)f_{g_0}(t))^2+C^2(f_{n,g_0}(t)-f_{g_0}(t))^2\\
            &&\quad +2C|(r_{n,g_0}(t)-m(t)f_{g_0}(t))(f_{n,g_0}(t)-f_{g_0}(t))|\big\}\\
            &\le&\frac{1}{f_{g_0}^2(t)}\left[(r_{n,g_0}(t)-m(t)f_{g_0}(t))^2+C^2(f_{n,g_0}(t)-f_{g_0}(t))^2\right.\\
            &&\quad \left.+C\{(r_{n,g_0}(t)-m(t)f_{g_0}(t))^2+(f_{n,g_0}(t)-f_{g_0}(t))^2\}\right]\\
            &=&\frac{1}{f_{g_0}^2(t)}\{C^\ast(r_{n,g_0}(t)-m(t)f_{g_0}(t))^2+C^{\ast\ast}(f_{n,g_0}(t)-f_{g_0}(t))^2\},
        \end{eqnarray*}
        where $C^\ast=1+C$ and $C^{\ast\ast}=C(1+C)$. Therefore,
        \begin{eqnarray*}
            &&\hskip-30pt E(m_{n,g_0}(t)-m(t))^2\\
&\le&\frac{1}{f_{g_0}^2(t)}\{C^\ast E(r_{n,g_0}(t)-m(t)f_{g_0}(t))^2+C^{\ast\ast}E(f_{n,g_0}(t)-f_{g_0}(t))^2\}.%\label{eq:thm:linearization_order:2}
        \end{eqnarray*}
        By substituting this expression in~\eqref{eq:thm:linearization_order:1}, we have
        \begin{eqnarray}
            &&\hskip-30pt E\left(|m_{n,g_0}(t)-m(t)-H_n(t)|\right)\notag\\
            &&\le C_1^\ast E(r_{n,g_0}(t)-m(t)f_{g_0}(t))^2+C_1^{\ast\ast}E(f_{n,g_0}(t)-f_{g_0}(t))^2, \label{eq:thm:linearization_order:2}
        \end{eqnarray}
        where
        \begin{eqnarray}
            C_1^\ast=\frac{1}{2f_{g_0}(t)}\frac{C^\ast}{f_{g_0}^2(t)};\quad  C_1^{\ast\ast}=\frac{1}{2f_{g_0}(t)}\left(\frac{C^{\ast\ast}}{f_{g_0}^2(t)}+1\right).\notag
        \end{eqnarray}

        We now have to calculate the orders of $E(r_{n,g_0}(t)-m(t)f_{g_0}(t))^2$ and $E(f_{n,g_0}(t)-f_{g_0}(t))^2$ on the right hand side of~\eqref{eq:thm:linearization_order:2}. It follows from the expression of $r_{n,g_0}(t)$ that
        \begin{eqnarray}
            E\left(r_{n,g_0}(t)\right)&=&\frac{n_1}{n}\int K(u)m(t-h_nu)f_1(t-h_nu)du\notag\\
                &&+\frac{n_2}{n}\int K(u)m(t-h_nu)f_2\comp g_0^{-1}(t-h_nu)(g_0^{-1})'(t-h_nu)du.\notag
        \end{eqnarray}
        The first order Taylor expansions of $m(t-h_nu)$, $f_1(t-h_nu)$, $f_2\comp g_0^{-1}(t-h_nu)$, and $(g_0^{-1})'(t-h_nu)$ around the values of these functions at $t$, together with Assumptions~A1$'$--A4$'$\ and the dominated convergence theorem, give
        \begin{eqnarray}
            E\left(r_{n,g_0}(t)\right)&=&\frac{n_1}{n}m(t)f_1(t) + \frac{n_2}{n} m(t)f_2\comp g_0^{-1}(t)(g_0^{-1}(t))' + o(h_n).\notag
        \end{eqnarray}
        which, by Assumption~A5$'$\ can be expressed as
        $$E(r_{n,g_0}(t)) - m(t)f_{g_0}(t) = o(h_n).$$
        Further, it transpires from the proof of Theorem~\ref{thm:pointwise_conv} (see discussion preceding \eqref{eq:VNn_conv}) that
       $$V(r_{n,g_0}(t)) = V(N_{n,t}(g_0)) = E(T_{1n})+E(T_{2n})-\frac1{n_1}E_{1n}^2-\frac1{n_2}E_{2n}^2,$$ where the first two terms are $O(\frac1{nh_n})$ and the last two are $O(\frac1n)$. Therefore,
       $$V(r_{n,g_0}(t)) = O\left(\frac1{nh_n}\right) = o\left(\frac{1}{\sqrt{nh_n}}\right).$$
        By putting the expressions of $E(r_{n,g_0}(t))$ and $V(r_{n,g_0}(t))$ together, we have
        \begin{eqnarray}
            E(r_{n,g_0}(t)-m(t)f_{g_0}(t))^2=o\left(h^2+\frac{1}{\sqrt{nh_n}}\right).\label{eq:thm:linearization_order:4}
        \end{eqnarray}
        Similarly, it follows from the expression of $f_{n,g_0}(t)$ that
        \begin{eqnarray}
            E(f_{n,g_0}(t)) - f_{g_0}(t) &=& o(h_n),\notag\\
            V(f_{n,g_0}(t))&=&o\left(\frac{1}{\sqrt{nh_n}}\right),\notag\\
                E(f_{n,g_0}(t)-f_{g_0}(t))^2&=&o\left(h_n^2+\frac{1}{\sqrt{nh_n}}\right).\label{eq:thm:linearization_order:3}
        \end{eqnarray}
        By combining~\eqref{eq:thm:linearization_order:2}--\eqref{eq:thm:linearization_order:3}, we have
        \begin{eqnarray}
            E\left(|m_{n,g_0}(t)-m(t)-H_n(t)|\right)=o\left(h_n^2+\frac{1}{\sqrt{nh_n}}\right).\notag
        \end{eqnarray}
        It follows from Markov's inequality that
        \begin{eqnarray}\label{eq:thm:linearization_order:main}
            |m_{n,g_0}(t)-m(t)-H_n(t)| =o_P\left(h_n^2+\frac{1}{\sqrt{nh_n}}\right).
        \end{eqnarray}
         From \eqref{eq:Hn2moment} and Chebyshev's inequality, we have $H_n(t)=O_P(b_n)$, while \eqref{eq:thm:linearization_order:main} can be rewritten as $m_{n,g_0}(t)-m(t)=H_n(t)+o_P(a_n)$, where
        $$a_n=h_n^2+\frac{1}{\sqrt{nh_n}};\quad b_n^2=h_n^4+\frac{1}{nh_n}.$$
        The above facts imply
        \begin{eqnarray*}
           &&\hskip-30pt (m_{n,g_0}(t)-m(t))^2\ =\ H_n^2(t) + o_P^2(a_n) + O_P(b_n)o_P(b_n)\\
            &=&H_n^2(t) + a_n^2o_P^2(1) + a_nb_nO_P(1)o_P(1)\
            =\ H_n^2(t) + a_n^2o_P(1) + a_nb_no_P(1).
        \end{eqnarray*}
        Consequently,
        $$MSE(m_{n,g_0}(t))=E\left(H_n^2(t)\right)+o(a_n^2+a_nb_n).$$
        Since $a_n^2 > b_n^2$, we have $a_n^2\ge a_nb_n$, and therefore
        $$o\left(a_n^2+a_nb_n\right) = o\left(a_n^2\right) = o\left(h_n^4+\frac{1}{nh_n}+h_n^2\frac{1}{\sqrt{nh_n}}\right) = o\left(h_n^4+\frac{1}{nh_n}\right).$$
        This finding, together with \eqref{eq:Hn2moment}, completes the proof.
    \end{proof}
       %---------- Proof of Theorem 3.4: MSE ----------%

%\section*{Supplementary material}\label{sec:supplement}
\bibliographystyle{agsm}
\bibliography{C:/dbhaumik/MyThesis/mybib}
\end{document}